\documentclass[10pt,journal,compsoc]{IEEEtran}
\ifCLASSOPTIONcompsoc
  \usepackage[nocompress]{cite}
\else
  \usepackage{cite}
\fi

\usepackage{cite}
\usepackage{amsmath,amssymb,amsfonts}
\usepackage{algorithm,algorithmic}
\usepackage{graphicx}
\usepackage{textcomp}
\usepackage{xcolor}
\usepackage{color}
\usepackage{amsthm}
\usepackage{balance}
\usepackage{multirow}
\usepackage{booktabs}
\usepackage{tabulary}
\usepackage{pifont}
\newcommand{\cmark}{\ding{51}}
\newcommand{\xmark}{\ding{55}}
\newcommand{\yifeng}{\textcolor{blue}}

\usepackage{enumerate}

\newtheorem{theorem}{\textbf{\emph{Theorem}}}
\newtheorem{definition}{\textbf{\emph{Definition}}}

\newcommand{\msf}{\mathsf}

\ifCLASSINFOpdf
\else
\fi

\begin{document}
\title{Aggregation Service for Federated Learning: An Efficient, Secure, and More Resilient Realization}

\author{Yifeng~Zheng, Shangqi Lai, Yi Liu, Xingliang Yuan, Xun Yi, and Cong Wang, \emph{Fellow, IEEE}
\IEEEcompsocitemizethanks{

\IEEEcompsocthanksitem Yifeng Zheng is with the School of Computer Science and Technology, Harbin Institute of Technology, Shenzhen, Shenzhen 518055, China. E-mail: yifeng.zheng@hit.edu.cn.

\IEEEcompsocthanksitem Shangqi Lai and Xingliang Yuan are with the Faculty of Information Technology, Monash University, Australia. E-mail: shangqi.lai@monash.edu, xingliang.yuan@monash.edu.

\IEEEcompsocthanksitem Yi Liu and Cong Wang are with the Department of Computer Science, City University of Hong Kong, Hong Kong. E-mail: 97liuyi@gmail.com, congwang@cityu.edu.hk.

\IEEEcompsocthanksitem Xun Yi is with the School of Computing Technologies, RMIT University, Australia. E-mail: xun.yi@rmit.edu.au.


}
}

\IEEEtitleabstractindextext{
\begin{abstract}
Federated learning has recently emerged as a paradigm promising the benefits of harnessing rich data from diverse sources to train high quality models, with the salient features that training datasets never leave local devices. Only model updates are locally computed and shared for aggregation to produce a global model. While federated learning greatly alleviates the privacy concerns as opposed to learning with centralized data, sharing model updates still poses privacy risks. In this paper, we present a system design which offers efficient protection of individual model updates throughout the learning procedure, allowing clients to only provide obscured model updates while a cloud server can still perform the aggregation. Our federated learning system first departs from prior works by supporting lightweight encryption and aggregation, and resilience against drop-out clients with no impact on their participation in future rounds. Meanwhile, prior work largely overlooks bandwidth efficiency optimization in the ciphertext domain and the support of security against an actively adversarial cloud server, which we also fully explore in this paper and provide effective and efficient mechanisms. Extensive experiments over several benchmark datasets (MNIST, CIFAR-10, and CelebA) show our system achieves accuracy comparable to the plaintext baseline, with practical performance.
\end{abstract}

\begin{IEEEkeywords}
Federated learning, secure aggregation, privacy, quantization, computation integrity
\end{IEEEkeywords}}

\maketitle

\IEEEdisplaynontitleabstractindextext

\IEEEpeerreviewmaketitle

\section{Introduction}

Federated learning has rapidly emerged as a fascinating machine learning paradigm \cite{McMahanMRHA17} which allows models to be trained on data dispersed over a number of mobile devices while each client can keep its dataset locally.
Clients never share the raw datasets, and instead only periodically share model updates locally trained with their datasets, which are then aggregated by a coordinating server to produce a global model.
Federated learning thus promises reaping the benefits of harnessing rich data from diverse sources to train high quality models, without clients being worried about the security and privacy risks of centralizing their raw data to a single place for training as in conventional practice.
%
%
%

While sharing model updates instead of raw datasets has greatly alleviated the privacy concerns as opposed to learning with centralized data, it still entails risks of private information leakage \cite{MelisSCS19}.
Protection of individual model updates is thus still necessary.
Additively homomorphic encryption, which allows addition of private plaintext values to be securely performed over their ciphertexts, offers a feasible solution. 
Some recent works using this technique for building privacy-preserving federated learning systems have been presented \cite{XuBZAL19,TruexBASLZZ19,PhongAHWM18,zhang2020batchcrypt}.
However, homomorphic encryption is expensive and also poses additional trust assumptions (e.g., sharing a private key across data holders) in the context of federated learning (see Section \ref{sec:related-work} for more discussion).
A different work is due to Bonawitz et al. \cite{BonawitzIKMMPRS17}, which is free of expensive cryptography and allows lightweight encryption and aggregation.

When designing secure federated learning systems, resilience against drop-out clients failing to submit model updates for aggregation is a practical requirement as some clients may face issues like poor network connections, energy constraints, or temporary unavailability \cite{LiSTS20}. 
The work \cite{BonawitzIKMMPRS17} provides mechanisms to handle drop-out clients, yet it undesirably prevents drop-out clients from directly and safely engaging in any future rounds of aggregation unless a new key setup is re-conducted.
%
Therefore, how to achieve secure and lightweight processing in federated learning while maintaining practical drop-out resilience with no impact on drop-out clients' participation in any future rounds remains to be fully explored.

In addition, most of prior works largely overlook the communication efficiency aspect of federated learning in the encrypted domain, which could be affected due to ciphertext size expansion or loss of precision in training due to the adaption of plaintext processing for compatibility with cryptographic processing.
%
%
Furthermore, most of them only provide semi-honest security against the passively adversarial server, assuming the server faithfully conduct the designated processing, and they are not resilient to an active adversary that may compromise the integrity of the processing at the server side.

In light of the above observations, in this paper, we present a new system design enabling federated learning services with lightweight secure and resilient aggregation.
Our system enables clients holding proprietary datasets to only provide obscured model updates while aggregation can still be supported at a cloud server to produce a global model.
By newly adapting a cherry-picked aggregation protocol for federated learning that we identified from the literature \cite{KursaweDK11,MelisDC16}, our system promises practical efficiency on encrypting the model updates at the client as well as aggregating the obscured model updates at the cloud server.
No expensive cryptographic operations are required throughout the federated learning procedure.
Meanwhile, our system is drop-out resilient and outperforms the best prior work \cite{BonawitzIKMMPRS17} in the sense that the secret keys of drop-out clients are kept private so they still can directly participate at a later stage without the need for a new key setup.

With the above new secure design point for federated learning as a basis, we explore and present refinements in terms of boosted communication efficiency and stronger security.
We start with consideration on the communication side, a known bottleneck for federated learning, due to various reasons like large sized models, limited client uplink bandwidth, and large numbers of participating clients \cite{McMahanMRHA17}.
We thus explore the potential of compressing model updates before secure aggregation so as to reduce the communication overheads in our system.

In the literature, various kinds of techniques for compressing the model updates have been proposed, such as sparsification, subsampling, and quantization \cite{LiSTS20}.
Our observation is that quantization delicately represents the (fractional) values in a model update as integers, which is also the type of data required for secure aggregation.
So our insight is to integrate the advancements in quantization with secure aggregation.
However, most of existing quantization schemes are not secure aggregation friendly, as  de-quantization requiring computation beyond summation has to be conducted before the quantized model updates can be aggregated.
We make an observation that a newly developed quantization technique \cite{zhang2020batchcrypt} (originally tailored for homomorphic encryption) can suit our purpose as no de-quantization is required before aggregation.
Building on this new technique, our system allows clients to provide obscured quantized model updates, achieving a reduction in the communication overhead.
To correctly integrate this quantization technique and make it function well, we address some practical considerations in our system including the prevention of overflow in aggregating quantized values and the identification of negative values in de-quantization.

Apart from the boosted communication efficiency side, we also investigate mechanisms to make our system more resilient, achieving stronger security against an actively adversarial cloud server, beyond the semi-honest security setting commonly assumed.
The goal here is to ensure that the processing for the federated learning service is correctly enforced at the cloud server.
Aiming for a pragmatic solution, we resort to the emerging techniques of hardware-assisted trusted execution environments (TEEs) to shield the computation integrity at the cloud server side throughout the federated learning procedure.
As a practical instantiation, our system makes use of the increasingly popular Intel SGX \cite{McKeenABRSSS13,intelsgx}.
We do not rely on the confidentiality guarantee which is originally targeted by trusted hardware, due to the emergence of various side-channel attacks \cite{TramerZLHJS17}.
This is different from most of existing works which assume both confidentiality and integrity when using trusted hardware for secure computation.
We give the abstract functionality assumed out of trusted hardware for our federated learning system, and further present a concrete protocol that renders federated learning with secure aggregation with computation integrity against an actively adversarial server.

We conduct extensive experiments over the popular benchmark datasets (MNIST, CIFAR-10, and CelebA), and different deep neural network models, with $21,840$ parameters, $23,272,266$ parameters, and $13,962,562$ parameters, respectively.
We evaluate the accuracy evolution over varying rounds and demonstrate our system shares similar behavior with the plaintext baseline and achieve comparable accuracy, even if quantization is applied to the model updates (which achieves up to $4\times$ reduction in communication).
The client-side security cost in encryption as well as in dealing with drop-out is thoroughly examined, which is on the order of a few seconds for the MNIST model and of a few minutes for larger CIFAR-10 and CelebA models.

Our evaluation on the server-side aggregation demonstrates that our system with security against an actively adversarial cloud server incurs almost no overhead over the semi-honest adversary setting.
We also make a performance comparison with the state-of-the-art \cite{BonawitzIKMMPRS17} (which exposes the secret keys of drop-out clients).
The results demonstrate that our system is (up to $39\times$) much more efficient under zero client drop-out and only incurs small computation overhead (limited to $2.3\times$) under varying drop-out rates.
Besides, our system (with a semi-honest cloud server) has less server computation as the drop-out rates and number of participating clients increase.

We highlight our contributions as follows:

\begin{itemize}

\item We present a new system design enabling federated learning services with lightweight secure and resilient aggregation. Compared to prior work, our system can handle client drop-out while keeping their secret keys confidential, so their direct participation in future rounds is not affected.

\item We explore quantization-based model compression and newly make unique integration with secure aggregation to boost the communication efficiency in our system. We also present practical mechanisms to make our system more resilient, achieving stronger security against an actively adversarial server.

\item We conduct extensive experiments over multiple real-world datasets and extensively evaluate the accuracy, client-side, and server-side performance. The results demonstrate that our system  achieves accuracy comparable to the plaintext baseline, with practical performance.
\end{itemize}

The rest of this paper is organized as follows.
Section \ref{sec:related-work} discusses the related work.  
Section \ref{sec:prelimianries} introduces some preliminaries.
Section \ref{sec:system_overview} gives a system overview.
Section \ref{sec:design_sec_agg} presents the design of secure aggregation in our system. 
Section \ref{sec:quantization design} shows how to integrate quantization to boost the communication efficiency.
Section \ref{sec:integrity-assued-design} gives the design of endowing our system with integrity.
Section \ref{sec:experiments} presents the experiments.
Section \ref{sec:conclusion} concludes the whole paper.

\section{Related Work}
\label{sec:related-work}


%
Our research is related to the line of work on federated learning with secure aggregation to protect the individual model updates.
To the best of our knowledge, none of existing works tackle exactly the same problem as our work, i.e., federated learning with secure and lightweight aggregation, drop-out resilience, and computation integrity against an adversarial server.

Some works rely on expensive homomorphic encryption \cite{TruexBASLZZ19,zhang2020batchcrypt,PhongAHWM18}, incurring high performance overheads.
The use of homomorphic encryption also requires all clients to either share a common secret key \cite{zhang2020batchcrypt,PhongAHWM18}, or hold secret key shares which have to be generated via expensive multi-party protocols or distributed by a trusted third party (TTP) \cite{TruexBASLZZ19}.
Xu et al. \cite{XuBZAL19} propose a scheme using functional encryption, requiring a TTP for setting up keys as well as getting involved in the learning process.
Our system is free of such a TTP.

In \cite{BonawitzIKMMPRS17}, Bonawitz et al. propose a lightweight encryption scheme supporting secure aggregation for federated learning. 
Their scheme can handle drop-out clients in an aggregation round, but would reveal their secret keys.
When the scheme is used in federated learning, drop-out clients in a certain round are thus not able to directly and safely participate in any future rounds \emph{unless} a new key setup is conducted.
%
%
%
We note that the design of \cite{BonawitzIKMMPRS17} has a double-masking mechanism, which is to prevent the server from learning the data of users who are too late in sending their masked vectors and are assumed to drop-out by the server (so the server recovers their secret keys).

There are some follow-up works \cite{abs-2002-04156,abs-2009-11248,abs-2012-05433} that attempt to the improve \cite{BonawitzIKMMPRS17} from the \emph{performance} aspect, yet they require additional assumptions regarding client topology information \cite{abs-2002-04156,abs-2012-05433} or delicate privacy-efficiency balance \cite{abs-2009-11248}.
Specifically, the work of So et al. \cite{abs-2002-04156} applies a multi-group circular strategy for model aggregation, which partitions clients into communication groups and operates under a multi-group communication structure that relies on the topology of users and leads to a number of execution stages that has to be conducted sequentially across the groups.
Each client in a group has to communicate with every client in the next group.
%
The work of Choi et al. \cite{abs-2012-05433} has to appropriately organize clients in a graph structure based on the specific topology information and uses the graph to represent how public keys and secret shares are assigned to the other clients, so each client is aware of its neighboring clients.
The work of Kadhe et al. \cite{abs-2009-11248} constructs a multi-secret sharing scheme and leverages it to design a secure aggregation scheme, which needs to delicately balance the trade-off between the number of secrets, privacy threshold, and dropout tolerance.
Meanwhile, the resilience against client dropout is restricted to some pre-set \emph{constant} fraction of clients.

The most related prior work thus is \cite{BonawitzIKMMPRS17}. Compared with \cite{BonawitzIKMMPRS17}, our system also allows lightweight encryption for the clients but does not reveal the secret keys of drop-out clients in any rounds.
That said, drop-out or non-selected clients in a certain round are still able to directly have safe participation in future rounds in our system.
It is noted that the protocol in \cite{BonawitzIKMMPRS17} consists of multiple rounds (including key setup), which need to be all executed over the clients selected in each iteration when applied to federated learning.
%
%
Given the possibility of client drop-out being considered, it may not be promising to assume that each selected client in an iteration will well get involved in those interactions for key setup. A new key setup in each iteration of federated learning thus might hinder that iteration in proceeding normally.
Indeed, once the number of dropped clients (with all interactions rounds considered in an iteration of federated learning) exceeds a threshold, the current iteration of federated learning has to abort and start over.
In contrast, our protocol has minimal interactions among the selected clients and the server in each iteration of federated learning.
In fact, as opposed to plaintext-domain federated learning, our protocol just needs one additional round of interaction for handling client drop-out upon secure aggregation.
It is further worth noting that compared with existing works, our system ambitiously and newly provides assurance on computation integrity against the cloud server.

\begin{table}[t!]
\centering
\small

\caption{Comparison of federated learning systems with secure aggregation. ``\textbf{L}''ightweight cryptography.``\textbf{R}''esilience against drop-out clients so that secure aggregation can still be correctly accomplished. ``\textbf{K}''eys kept secret in handling dropouts so that drop-out clients' secret keys are not exposed.  ``\textbf{S}''ecurity for computation integrity against the server. ``\textbf{C}''ompression and security co-design for high communication efficiency while ensuring security.  ``\textbf{I}''ndividual keys without TTP.}
\begin{tabular}{@{}cllllll@{}}
\toprule
\textbf{System} & \textbf{L} & \textbf{R} & \textbf{K} & \textbf{S} & \textbf{C} & \textbf{I} \\ \midrule
Xu et al. \cite{XuBZAL19} & \xmark & \cmark & \cmark & \xmark& \xmark &\xmark \\
Phong et al. \cite{PhongAHWM18}       & \xmark   &\xmark   & \cmark  & \xmark  & \xmark   & \xmark  \\
Truex et al. \cite{TruexBASLZZ19} & \xmark &\cmark &\cmark &\xmark &\xmark &\xmark \\
Zhang et al. \cite{zhang2020batchcrypt}       & \xmark   &\cmark   & \cmark  & \xmark  & \cmark   & \xmark  \\
Mandal et al. \cite{MandalG19} & \xmark & \cmark & \cmark &\xmark &\xmark & \xmark \\
Bonawitz et al. \cite{BonawitzIKMMPRS17}       & \cmark  & \cmark  & \xmark  & \xmark  & \xmark  & \cmark  \\ 
\textbf{This work}       & \cmark  & \cmark  & \cmark  & \cmark  &\cmark   & \cmark  \\ \bottomrule
\end{tabular}
\label{table:comparison}
\end{table}

In an independent work, Mandal et al. \cite{MandalG19} consider a special scenario where they hide the model from the clients and rely on homomorphic encryption to build specific federated linear/logistic regression models. 
Our system follows most prior works with focus on protecting individual model updates and is general for training any models.
We note that some previous works \cite{TruexBASLZZ19,XuBZAL19} also integrate orthogonal differential privacy techniques by adding calibrated noises to local model updates and lead to an inherent trade-off on accuracy and privacy with delicate parameter tuning, which is complementary to our system.
Table \ref{table:comparison} summarizes the comparison with the most related works discussed above.

Our work is also relevant to prior work on building hardware-assisted security applications.
Most of existing works assume confidentiality and integrity guarantees from the trusted hardware (e.g., Intel SGX) for secure computation (e.g., \cite{SchusterCFGPMR15,OhrimenkoSFMNVC16,BahmaniBBPSSW17,ChoiTHPMSBT19,TramerB19}, to just list a few).
Given that the confidentiality guarantees may face threats from various side-channel attacks, a trending practice is to relax the trust assumptions and only assume integrity.
A few works have been presented in different contexts including zero-knowledge proofs \cite{TramerZLHJS17}, multi-party machine learning inference \cite{Kumar2020}, client-side checking of model training \cite{ZhangLZLWW20}, and aggregation over blockchain-stored data under homomorphic encryption \cite{DuanZDZWA19}.
Inspired by this trend, our system only assumes minimally trusted hardware with integrity guarantees, and explores a new design point for enforcing server-side computation correctness throughout the federated learning procedure.

\section{Preliminaries}
\label{sec:prelimianries}
\subsection{Federated Learning}

Federated learning enables multiple data owners to jointly solve an optimization problem which could be formulated as: $\min \sum\nolimits^S_{s=1}{\frac{1}{S}\cdot L(\mathbf{w},\mathcal{D}_s)}$, where $S$ is the number of data owners, $L(\mathbf{w},\mathcal{D}_s)$ is a loss function capturing how well the parameters $\mathbf{w}$ (treated as a flattened vector) model the local dataset $\mathcal{D}_s$.
During the learning procedure, each data owner only shares a locally trained model update.
The model updates are typically aggregated by a server and used to update the global model.
This is an iterative procedure and runs in multiple rounds.

Each round proceeds through the following steps: (1) A fraction of the data owners (say $K$ data owners) is selected by the server and a current global model $\mathbf{w}$ is sent to these data owners.
(2) Each selected data owner then performs training over its local dataset, for which any optimizers could be used, though stochastic gradient descent (SGD) is the most popular one.
With SGD, the $k$-th selected data owner updates the local model parameters via $\mathbf{w}_k \leftarrow \mathbf{w}-\eta\cdot \nabla L(\mathbf{w};\beta)$, where $\beta$ is a batch randomly sampled from the local dataset and $\eta$ is the learning rate.
A full iteration over the whole local dataset is referred to as an epoch, and the local training could be performed over multiple epochs.
(3) Once the local training is done, an individual model update $\mathbf{w}_k$ is shared for aggregation: $\mathbf{w}=\frac{1}{K}\sum\nolimits^K_{k=1}{\mathbf{w}_k}$, which produces an updated global model for next round.
In Table \ref{tbl:notations}, we provide a summary of the main notations used in this paper.

\begin{table}[t!]
\centering
\small
\caption{Key Notations}

\begin{tabular}{cl}
\toprule
Notation      & Description                                        \\ \midrule
$\mathbf{w}$    & Aggregate model      \\ 
$\mathbf{w}^t_k$          & Model update from client $\mathcal{C}_k$ in round $t$              \\ 
$(SK_i,PK_i)$          & Key pair of client $\mathcal{C}_i$ in secure aggregation\\ 
$CK_{i,j}$          & Shared key computed from $\mathbf{SK}_i$ and $\mathbf{PK}_j$ \\ 
$\mathcal{T}$          & Set of selected clients in each round \\ 
$\mathcal{T}_{o}$          & Set of online clients in each round \\ 

$\mathcal{T}_{d}$          & Set of drop-out clients in each round \\ 

$\mathbf{r}_{k}$          & The vector of blinding factors of client $\mathcal{C}_k$ \\ 
$Q_{r}$          & The $r$-bit quantizer used in our system \\ 
$Q^{-1}_{r}$          & The $r$-bit de-quantizer used in our system \\ 
$(sk_{\mathcal{C}_i,},pk_{\mathcal{C}_i})$          & Signing key pair of client $\mathcal{C}_i$ \\ 
\bottomrule 
\end{tabular}
\label{tbl:notations}
\end{table}


%

\subsection{Transparent Enclave}

Trusted hardware techniques (e.g., Intel SGX \cite{intelsgx,McKeenABRSSS13}) allow the creation of protected memory regions called enclaves which are isolated from the rest of a host's software, including the operating system.
Trusted hardware is designed with the aim of offering confidentiality and integrity assurance.
Yet, it has been recently shown that the confidentiality guarantee may be compromised by a series of side-channel attacks \cite{TramerZLHJS17}.
The goal of achieving confidentiality with trusted hardware thus remains elusive so far.
A recent trend on building trusted hardware based security applications \cite{TramerZLHJS17,Kumar2020} is thus to only leverage the computational integrity.
By relying on only the integrity, it is assumed that the internal states (i.e., all the code and data) of enclaves are transparent to the host (or the corrupted party).
Only secure code attestation and secure signing functionality are required. 
Enclaves with such minimal trust assumption are referred to as transparent enclaves \cite{TramerZLHJS17}.
%

\section{System Overview}
\label{sec:system_overview}
\subsection{Architecture}

\begin{figure}[t!]
\centerline{\includegraphics[width=0.43\textwidth]{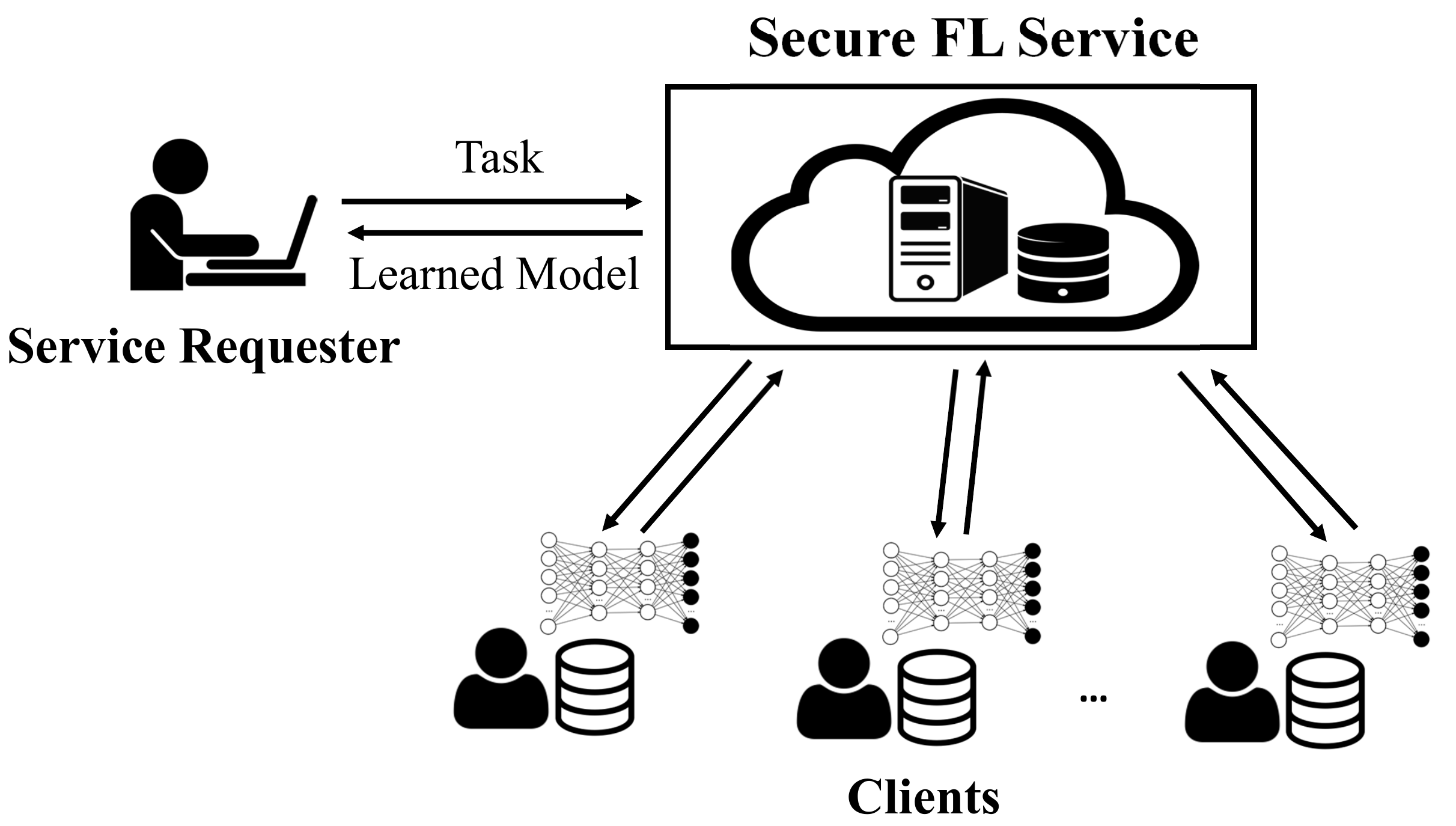}}
\caption{Overview of our system architecture.}
\label{fig:system_architecture}
\end{figure}

Fig. \ref{fig:system_architecture} illustrates our system architecture, which involves three actors: the service requester (or \emph{requester} for short), clients, and the secure federated learning service which bridges the requester and clients.
The requester wants to harness the power of federated learning and releases a task via the secure federated learning service.
Such service could be hosted on the cloud, and run by a cloud server.
Here the cloud server is an abstraction and can be implemented by an actual server or a cluster of servers.
The clients are data owners interested in working on the federated learning task and sign up at the cloud server.
All participants in the system agree in advance on a common machine learning model architecture and common learning objective.
The clients may receive rewards from the requester for the participation.

To realize the secure federated learning service, our first aim is that in each round the cloud server should be able to perform aggregation of the individual model updates without seeing them in the clear.
So, we will craft our design such that each client selected in a round can just provide an obscured model update, which still effectively allows aggregation.
For high efficiency, our system will preclude the use of expensive homomorphic encryption and only rely on lightweight cryptography.
Considering that communication is a known bottleneck in federated learning, our system will also cherry-pick and integrate appropriate quantization techniques to compress the model updates, so that a client can just share an obscured and quantized model update.

In addition to the confidentiality protection of individual model updates in each round of federated learning, our system also ambitiously aims to \emph{efficiently} ensure that the federated learning service is correctly provisioned by the cloud server. 
While theoretically (expensive) cryptography-based verifiable computation techniques may help \cite{XuLL0L20}, we aim for a pragmatic solution and observe the trending practice is to leverage hardware-assisted trusted execution environments (TEE) \cite{TramerZLHJS17}.
Our system follows this trend and only assumes a minimally trusted hardware with integrity guarantees.
Specifically, we will rely on a transparent enclave to shield the computation integrity at the cloud server.

\subsection{Threat Model}

Our system for secure federated learning considers an adversarial cloud server.
It could be semi-honest, which means the cloud server will honestly follow the protocol for the federated learning service, yet is curious about individual model updates so as to infer clients' local datasets.
The semi-honest threat model is commonly adopted in privacy-preserving data-centric services in cloud computing \cite{ZhengDTWZ21,LiuZYY21}.
For this setting our system aims to maintain the confidentiality of individual model updates. The cloud server is only allowed to learn an aggregate model update and global model.
Beyond the semi-honest adversary setting, a stronger threat model will also be considered where the cloud server may not correctly follow the designated computation.
For this setting, our system further aims to enforce correct computation at the cloud server besides the confidentiality of individual model updates.

Collusion between the cloud server and a subset of selected clients in each round is also likely, and for such case our system aims to maintain the confidentiality of model updates of honest clients.
Here, we are not concerned with the extreme (and also hardly realistic given that clients are many and geographically distributed) case that there is only one honest client and all other clients are compromised by the cloud server.
Note that the assumption of at least one honest client is necessary in the problem of secure data aggregation and also (either explicitly or implicitly assumed) in existing works \cite{LiCP14,ZhangCZ17,MelisDC16,BonawitzIKMMPRS17}.
It is obvious that for any practical secure aggregation schemes that securely compute the sum, if the aggregator colludes with all but one of the clients, the data of the colluding clients can be surely removed from the sum and the honest client's data will be revealed.
Therefore, like existing works, our system does not provide data protection against the extreme colluding case.

It is noted that our system currently does not protect against possible attacks on the aggregate model updates and global model, which can directly be mitigated by letting clients do noise addition as per complementary differential privacy techniques \cite{TruexBASLZZ19}.
Adversarial attacks like poisoning attacks \cite{FangCJG20} and backdoor attacks \cite{BagdasaryanVHES20} are complementary research areas and out the scope of this work.

\section{Empowering Federated Learning with Secure and Efficient Aggregation}

\label{sec:design_sec_agg}
\subsection{Overview}

Our system is aimed at federated learning service with secure aggregation that is free of heavy cryptography.
Meanwhile, our system should be able to work without a trusted third party for setting up keys for secure aggregation and/or assistance in the learning procedure.
All these practical requirements preclude considerations on the use of expensive homomorphic encryption which is also confronted with key distributions issues as mentioned before.
Recall that although the secure aggregation scheme in \cite{BonawitzIKMMPRS17} allows lightweight encryption and aggregation without a trusted third party, it reveals the secret keys of drop-out clients in an aggregation round, which consequently prevents them from direct participation in future rounds unless a new key setup is conducted in a subsequent round.

In our system, we propose to efficiently protect the confidentiality of individual model updates in federated learning with a cherry-picked low overhead cryptographic aggregation protocol \cite{MelisDC16}.
We resort to this protocol due to the observations that it does not involve heavy cryptographic operations (only lightweight hashing and arithmetic operations are needed), does not require a trusted third party, and will not reveal the secret keys of drop-out clients.
With this protocol as a basis, we craft a thorough design for federated learning with secure aggregation of individual model updates.
Later in Section \ref{sec:quantization design} and Section \ref{sec:integrity-assued-design}, we will work through the refinements that promise boosted communication efficiency and security.

\subsection{Protocol}

We now present the protocol for federated learning with secure aggregation of model updates.
Note that as required by the cryptographic computation, the values in the vector representing the local model update should integers, which, however, could be fractional values in practice. 
%
%
We adopt a common scaling factor trick where a large enough scaling factor is used to scale up a fractional value into an integer.
Specifically, given a fractional value $v$ and a scaling factor $L$, we can obtain an integer representation of $v$ as $\overline{v}= \left\lfloor v \cdot L \right\rfloor$.
The approximate $v$ can be later reconstructed as $\overline{v}/L$ \cite{WangRWW13}.
Applying such trick in our context, we only need to scale down the aggregate model update.
To ensure that the scaling operation does not compromise the quality of computation result, the message space for the scaled integers should be large enough (say $2^{32}$ or $2^{64}$).
Our protocol, as shown in Algorithm \ref{algo:sec-fl-algo}, proceeds as follows.

\noindent\textbf{Initialization.}
Suppose that there are totally $S$ clients in the system.
We write $\mathcal{S}$ to denote the set of clients, in which each client is uniquely indexed by an integer $i\in[1,S]$.
For initialization, each client generates its own key pair and leverages the cloud server as a central hub for distribution of public keys.
Let $\mathbb{G}$ be a cyclic group of prime order $p$, with generator $g$.
Also, let $\mathsf{H} $$: $$\{0,1\}^{*} $$\rightarrow $$\mathbb{Z}_p$ be a cryptographic hash function mapping arbitrary-length strings to integers in $\mathbb{Z}_p$.
Each client $\mathcal{C}_i$ generates a private key $SK_i=x_i \in \mathbb{Z}_p$ and a public key $PK_i$$=$$ g^{x_i}\in \mathbb{G} $, and uploads $PK_i$ to the cloud server. 
Then, each client $\mathcal{C}_i$ downloads other users' public keys and computes $CK_{i,j}=\mathsf{H}((PK_j)^{x_i})$, where $i, j\in \mathcal{S}$ and $j\ne i$. 
%
%

\setlength{\textfloatsep}{5pt}
\begin{algorithm}[!t]
\caption{Our Design for Secure Federated Learning}
\label{algo:sec-fl-algo}
\begin{algorithmic}[1]
\STATE \textbf{\underline{Initialization:}}

\FOR{each client $\mathcal{C}_i$}

\STATE $SK_i=x_i \mathop \in \mathbb{Z}_p$.
\STATE $PK_i$$=$$ g^{x_i}\in \mathbb{G}$.
\STATE Upload $PK_i$ to the cloud server.
  \FOR{each $j\in \mathcal{S},j\neq i$} 
    \STATE Download $PK_j$ from the cloud server.
    \STATE Compute  $CK_{i,j}=\mathsf{H}((PK_j)^{x_i})$.
  \ENDFOR
\ENDFOR

\STATE \textbf{\underline{Cloud server executes:}}
\STATE Initialize $\mathbf{w}^0$.

\FOR{each round $t=1,2,\cdots$}

  \STATE Select a fraction of the clients and produce the set $\mathcal{T}^{t}$.

  \FOR{each client $\mathcal{C}_k$ in $\mathcal{T}^{t}$} 
    \STATE $\mathbf{w'}^{t}_k$ $\leftarrow$ ClientUpdate($\mathbf{w}^{t-1}$,$\mathcal{T}^t$).
  \ENDFOR

  \STATE $\mathbf{w'} \leftarrow \sum\nolimits_{k\in\mathcal{T}^t} {\mathbf{w'}^t_k} \bmod M$ .
  \STATE $\mathbf{w}^t \leftarrow \frac{1}{|\mathcal{T}^{t}|} \cdot\mathbf{w'}$.
\ENDFOR

\vspace{2pt}

\STATE \textbf{ClientUpdate($\mathbf{w}^{t-1}$,$\mathcal{T}^t$):}

\vspace{2pt}

\STATE $\mathbf{w}\leftarrow \mathbf{w}^{t-1}$.


\STATE Split the dataset $\mathcal{D}_k$ into batches $\mathcal{B}$.

\FOR{each local epoch $e$ from $1$ to $E$}
  \FOR{each batch $\beta$ in $\mathcal{B}$} \STATE $\mathbf{w}\leftarrow \mathbf{w}-\eta\cdot \nabla L(\mathbf{w};\beta)$. // $\eta$ is the learning rate.  \ENDFOR
\ENDFOR

\STATE $\mathbf{w}^t_k\leftarrow \mathbf{w}$.

\STATE Generate the vector $\mathbf{r}^t_k$ of blinding factors. Each element $\mathbf{r}^t_k(b)$$=$$\sum\nolimits_{n \in\mathcal{T},n \ne k} { (-1)^{k>n} \mathsf{H}(CK_{k,n}||b||t)}$.

\STATE Compute $\mathbf{w'}^t_k\leftarrow \mathbf{w}^t_k+\mathbf{r}^t_k\bmod M$.

\STATE Return $\mathbf{w'}^t_k$ to the cloud server.
\end{algorithmic}

\end{algorithm}

\noindent\textbf{Secure federated learning.}
Without loss of generality, we describe here the secure aggregation process for federated learning in one round.
The cloud server first randomly selects a fraction $\eta$ of the clients.
We write $\mathcal{T}$ to denote the set of selected clients.
This set $\mathcal{T}$ and the current global model vector $\mathbf{w}$ ($\mathbf{w}$ is an initialized model when it is the first round) is sent to the selected clients.
%
%
%
Each selected client $\mathcal{C}_k$ ($k\in \mathcal{T}$) then performs local training and produces a model update $\mathbf{w}_k$.
An obscured model update is then generated based on blinding factors.
In particular, it first generates a blinding factor for each element in the model update vector $\mathbf{w}_k$ and produces a vector $\mathbf{r}_k$ of blinding factors.

The basic idea for supporting correct aggregation over obscured model updates is that if a client adds randomness to its input for blinding while another client subtracts that randomness from its input, the randomness will be canceled out when the summation of clients' obscured inputs is formed. 
In particular, each blinding factor $\mathbf{r}_k(b)$ for the $b$-th element in $\mathbf{w}_k$ is generated as $\mathbf{r}_k(b)$$=$$\sum\nolimits_{n \in\mathcal{T},n \ne k} { (-1)^{k>n} \mathsf{H}(CK_{k,n}||b||t)}$, where $(-1)^{k>n} $$=$$-1$ if $k$$>$$ n$ and $1$ otherwise, and $t$ is a round counter.
It can be observed that the sum of all blinding factors $\sum\nolimits_{k \in \mathcal{T}} {\mathbf{r}_k(b)}$ is $0$ as the blinding factors are canceled out when the sum over the set $\mathcal{T}$ of selected clients is formed.

Given above, with a model update $\mathbf{w}_k$ where each element is assumed to lie in a message space $M$ (say $M=2^{32}$), each client $\mathcal{C}_k$ generates a vector $\mathbf{r}_k$ of blinding factors as introduced above, and computes $\mathbf{w}'_k=\mathbf{w}_k+\mathbf{r}_k \bmod M$, where the modulo operation is performed element-wise.
Upon receiving the obscured model updates, the cloud server computes $\sum\nolimits_{k\in\mathcal{T}} {\mathbf{w}'_k} \bmod M=\sum\nolimits_{k\in\mathcal{T}} {(\mathbf{w}_k+\mathbf{r}_k)} \bmod M$ and obtains $\sum\nolimits_{k\in\mathcal{T}} {\mathbf{w}_k} \bmod M$ as $\sum\nolimits_{k \in \mathcal{T}} {\mathbf{r}_k}=0$.

\subsection{Security Guarantees}

We now provide security analysis for our above design below.
\begin{definition}

(Computational Diffie-Hellman (CDH) Problem). Consider a cyclic group $\mathbb{G}$ of prime order $p$ with generator $g$. The CDH problem is hard if, for any probabilistic polynomial time algorithm $\mathcal{A}$ and random $a$ and $b$ drawn from $\mathbb{Z}_p$: $\text{Pr} [\mathcal{A}(\mathbb{G};p; g; g^a; g^b) = g^{ab}]$ is negligible.
\end{definition}
\begin{theorem}

Given the hardness of the CDH problem, our system ensures that the cloud server only learns the aggregate model update without knowing individual model updates. In case of collusion between the cloud server and a subset of clients, the model updates of honest clients are still protected.
\end{theorem}

\begin{proof}

Each value in the model update of the client $\mathcal{C}_k$ is masked by a unique blinding factor generated by the secret keys $\{CK_{k,n}: \mathsf{H}((g^{x_n})^{x_k})\}$.
So we need to show that the cloud server is oblivious to the secret key $CK_{k,n}$.
In our system, the cloud server only has access to the public key $g^{x_i}$ of each client in $\mathcal{S}$, as per the system setup.
The CDH problem ensures that given $g^a$ and $g^b$, it is computationally hard to compute the value $g^{ab}$.
So, given access to the public keys $g^{x_k}$ and $g^{x_n}$, the cloud server is not able to infer the secret key $CK_{k,n}=\mathsf{H}((g^{x_n})^{x_k})$ for the generation of blinding factors.
Next, we analyze the case of passive collusion between a subset of selected clients with the cloud server, where the corrupted clients share all their secret materials with the cloud server.
%

Without loss of generality, we give the proof for a certain honest client denoted by $\mathcal{C}_i$.
Let $\mathcal{E}$ denote the set of clients participating in the same aggregation round $t$, $\mathcal{E}_{h}$ denote the set of honest clients, and $\mathcal{E}_{c}$ the subset of clients colluding with the cloud server.
Recall that in our system the data submitted by the honest client $\mathcal{C}_i$ is $\mathbf{w}_i(b)+\mathbf{r}_i(b)$$=$$\mathbf{w}_i(b)+\sum\nolimits_{n \in\mathcal{E},n \ne i} { (-1)^{i>n} \mathsf{H}(CK_{i,n}||b||t)}$, for the $b$-th element in $\mathbf{w}_i$.
Given $\mathcal{E}_h$ and $\mathcal{E}_c$, $\sum\nolimits_{n \in\mathcal{E},n \ne i} { (-1)^{i>n} \mathsf{H}(CK_{i,n}||b||t)}$ can be expressed as two parts: $\sum\nolimits_{n \in\mathcal{E}_{h},n \ne i} { (-1)^{i>n} \mathsf{H}(CK_{i,n}||b||t)}$ and $\sum\nolimits_{n \in\mathcal{E}_c} { (-1)^{i>n} \mathsf{H}(CK_{i,n}||b||t)}$.
So what the cloud server receives from the (honest) client is $\mathbf{w}_i+\sum\nolimits_{n \in\mathcal{E}_{h},n \ne i} { (-1)^{i>n} \mathsf{H}(CK_{i,n}||b||t)}+\sum\nolimits_{n \in\mathcal{E}_c} { (-1)^{i>n} \mathsf{H}(CK_{i,n}||b||t)}$.
As the set $\mathcal{E}_c$ of clients colludes with the cloud server, they are able to reveal to the cloud server the part $\sum\nolimits_{n \in\mathcal{E}_c} { (-1)^{i>n} \mathsf{H}(CK_{i,n}||b||t)}$.
So given the collusion with the set $\mathcal{E}_c$ of clients, the cloud server can obtain $\mathbf{w}_i+\sum\nolimits_{n \in\mathcal{E}_{h},n \ne i} { (-1)^{i>n} \mathsf{H}(CK_{i,n}||b||t)}$ ultimately.
It can be seen that the honest client $\mathcal{C}_i$'s model update is still protected by blinding factors generated from the mutual secret keys established between this honest client and other honest clients in $\mathcal{E}_h$.
That said, the cloud server still observes a randomly masked model update from the honest client $\mathcal{C}_i$.
This completes the proof.

\end{proof}

\subsection{Practical Considerations}

\noindent\textbf{Fault tolerance.}
It is likely that a selected client in a round may not actually participate (i.e., submitting a model update), due to reasons like poor network connections, energy constraints, or temporary unavailability.
We call such clients drop-out clients.
%
%
In such case, the cloud server would not be able to directly obtain the correct aggregation result, because the sum of the blinding factors is not zero.
%
Therefore, our system should be able to handle the drop-out clients in a certain round and ensure that the obscured model updates of the responding clients can still be correctly aggregated.

We assume that the selected clients in a round which have managed to submit the (obscured) model updates will be able to stay online for potential assistance.
The rationale behind such assumption is that clients will typically train and participate when their devices are charging and on an unmetered network \cite{McMahanMRHA17,BonawitzEGHIIKK19}.
In practice, adequate incentive mechanisms could also be developed to further motivate clients' active participation.

To handle drop-out clients, the main insight is that the blinding factors of the responding clients should be eventually canceled out.
%
%
Let $\mathcal{T}_{o}$ denote the set of responding clients who have submitted the (masked) model updates and are able to stay online for potential assistance, and $\mathcal{T}_{d}$ the set of drop-out clients who fail to submit the masked model updates.
At a high level, each online client $\mathcal{C}_k \in \mathcal{T}_{o}$ generates a vector consisting of blinding factors that are generated based on the shared key $\{CK_{k,n}\}_{n\in \mathcal{T}_{d}}$ with the set of drop-out clients.
This vector of blinding vectors allows the server to eliminate the blinding factors corresponding to the drop-out clients in each masked model update, and thus the aggregate model update.
We show in Algorithm \ref{algo:fault-tolerance} the fault tolerance mechanism for handling drop-out clients.

\noindent\textbf{Client dynamics.} 
After the initial system setup, new clients may join and some clients may be revoked later. 
%
%
We now introduce how our system can support such client dynamics.
%
%
When a new client $\mathcal{C}_i$ joins, it generates a secret key $x_i$ and uploads its public key $PK_i$$=$$g^{x_{i}}$ to the cloud server, which then broadcasts it to other current clients in the system.
Upon receiving $PK_i$, each client $\mathcal{C}_j$ computes $CK_{j,i}$$=$$(PK_i)^{x_j}$ to ensure that its blinding factor will correctly be generated for later use.
On another hand, when a client $\mathcal{C}_i$ leaves the system, the cloud server only needs to inform each client $\mathcal{C}_j$ to discard the key $CK_{j,i}$.

\begin{algorithm}[!t]
\caption{Fault Tolerance against Client Drop-out}
\label{algo:fault-tolerance}
\begin{algorithmic}[1]



\STATE Cloud server sends $\mathcal{T}_{d}$ to the clients in $\mathcal{T}_{o}$.

\FOR{each client $\mathcal{C}_k$ in $\mathcal{T}_{o}$}

\STATE Initialize a vector $\mathbf{q}_k$ of size $|\mathbf{w}'_k|$.
\FOR{the $b$-th element in $\mathbf{q}_k$}
\STATE $\mathbf{q}_k(b) $$\leftarrow$$ \sum\nolimits_{{n \in \mathcal{T}_{d}}}{  (-1)^{k>n} \mathsf{H}(CK_{k,n}||b||t)}\bmod M$.
\ENDFOR

\STATE Send $\mathbf{q}_k$ to the cloud server.
\ENDFOR

\STATE Cloud server computes $\mathbf{q}=\sum\nolimits_{k \in \mathcal{T}_{o} } {\mathbf{q}_k}\bmod M$.

\STATE Cloud server computes $(\sum\nolimits_{k \in \mathcal{T}_{o}}\mathbf{w}'_k)-\mathbf{q}\bmod M$.

\end{algorithmic}

\end{algorithm}

\noindent\textbf{Group management.}
In each round of federated learning, the client-side computation complexity in secure aggregation scales linearly in the size of the set $\mathcal{T}$, as generating the blinding factors requires each client to perform $O(|\mathcal{T}|)$ hashing operations. 
When $|\mathcal{T}|$ becomes quite large, it would be beneficial to improve the client side efficiency.
One immediate optimization is that in each round, the selected clients can be divided into groups with sizes smaller than $\mathcal{T}$.
During the process of secure aggregation, within-group aggregation is first performed at the cloud server, followed by cross-group aggregation.

Suppose that the selected clients are divided into $s$ groups with equal sizes, the computation complexity is reduced to $O(|\mathcal{T}|/s)$ for each client participating in a round of federated learning.
For client grouping, different guidelines might be adopted, such as geographical proximity. 
Note that group management also benefits fault tolerance. 
This is because each group is independent so the faults in a certain group will not affect other groups.
Only the clients in the same group will be involved to handle the faults.
The impact of faults on the overall system is thus effectively mitigated. 
One trade-off here is that the cloud server learns aggregated model updates from clients in independent groups. Yet, individual clients' model updates in each group are still protected. 
The group size can be set in advance according to the agreement between the cloud server and the clients.

\section{Leveraging Quantization for Boosted Communication Efficiency}
\label{sec:quantization design}

Communication could be a bottleneck for federated learning \cite{LiSTS20}, because the trained machine learning models are usually of large sizes, clients may have limited uplink bandwidth, and the number of clients could be large \cite{McMahanMRHA17}.
%
%
As in each round a selected client communicates the model update to the cloud server, reducing the model update size could be highly beneficial to improve the communication efficiency of the whole system.
In the literature, researchers have studied various kinds of techniques for compressing model updates such as sparsification, subsampling, and quantization \cite{LiSTS20}.
Among these techniques, our observation is that quantization delicately represents the values in a model update as integers, exhibiting compatibility with cryptographic aggregation. 
Therefore, we choose to leverage the advancements in quantization to support secure aggregation while achieving a reduction in communication cost, and thus kill two birds with one stone.

\subsection{Quantization Technique}

Quantization based communication efficiency optimization for distributed machine learning has received considerable traction in recent years (e.g., \cite{ReisizadehMHJP20,AlistarhG0TV17,ShlezingerCEPC20}, to just list a few).
However, most of existing quantization techniques are not secure aggregation friendly, as they require a de-quantization operation, which demands computation much more complicated than secure aggregation, to be performed before doing the aggregation.
What we need here is thus a quantization scheme which can support aggregation directly over quantized values.

\begin{algorithm}[!t]
\caption{Secure Federated Learning with Quantization}
\label{algo:fl-secagg-quant}
\begin{algorithmic}[1]

\STATE \textbf{\underline{Cloud server executes:}}
\STATE Initialize $\mathbf{w}^0$.

\FOR{each round $t=1,2,\cdots$}

  \STATE Select a fraction of the clients and produces the set $\mathcal{T}^{t}$.

  \FOR{each client $\mathcal{C}_k$ in $\mathcal{T}^{t}$} 
    \STATE $\mathbf{w'}^{t}_k$ $\leftarrow$ ClientUpdate($\mathbf{w}^{t-1}$,$\mathcal{T}^t$).
  \ENDFOR

  \STATE $\Delta^t \leftarrow \sum\nolimits_{k\in\mathcal{T}^t} {\mathbf{w'}^{t}_k} \bmod 2^r$.
  
  \STATE $\mathbf{w}^t=\mathbf{w}^{t-1}+ \frac{1}{|\mathcal{T}^{t}|}\cdot Q^{-1}_r(\Delta^t)$.
\ENDFOR

\vspace{10pt}

\STATE \textbf{ClientUpdate($\mathbf{w}^{t-1}$,$\mathcal{T}^t$):}

\vspace{2pt}

\STATE Produce a locally trained model $\mathbf{w}^t_k$ as in Algorithm 1 (steps 22-29).

\STATE $\Delta^t_k=\mathbf{w}^t_k-\mathbf{w}^{t-1}$.
\STATE Generate the vector $\mathbf{r}^t_k$ of blinding factors.

\STATE Compute $\mathbf{w'}^t_k\leftarrow Q_r(\Delta^t_k)+\mathbf{r}^t_k\bmod 2^r$.

\STATE Return $\mathbf{w'}^t_k$ to the cloud server.
\end{algorithmic}

\end{algorithm}

We make an observation that a newly developed quantization technique \cite{zhang2020batchcrypt} (originally tailored for homomorphic encryption) suits our purpose and has good compatibility with secure aggregation in our system.
This technique quantizes fractional values into $r$-bit signed integers.
In order to enable the values with opposite signs to be canceled out during summation, it proposes to make the quantized range symmetrical with respect to the range of the fractional values.  
Specifically, the $r$-bit quantizer $Q_r$ that we integrate in our system works as follows.
We start with introducing some notations.
For any scalar $v\in \mathbb{R}$, we write $\mathsf{sgn}(v)\in\{-1,1\}$ to denote the sign of $v$, with $\mathsf{sgn}(0) = 1$.
We write $\mathsf{abs}(v)$ to denote the absolute value of $v$, and $\mathsf{round}(v)$ to denote standard rounding over $v$.
Given a value $v$ in the range $[-B,B]$, the quantizer $Q_r$ quantizes it to an integer in $[-(2^{r-1}-1),2^{r-1}-1]$ via:
\begin{equation}
Q_r(v)=\mathsf{sgn}(v)\cdot \mathsf{round}(\mathsf{abs}(v)\cdot(2^{r-1}-1)/B)
\end{equation}


Given a quantized value $u=Q_r(v)$, the de-quantized value is computed as follows:
\begin{equation}
Q^{-1}_r(u)=\mathsf{sgn}(u)\cdot (\mathsf{abs}(u)\cdot B /(2^{r-1}-1))
\end{equation}


\subsection{Secure Federated Learning with Quantization}

We adapt the above quantization technique for our system to simultaneously achieve lightweight cryptographic aggregation and high communication efficiency.
Such delicate integration leads to our refined protocol for secure federated learning, which is shown in Algorithm \ref{algo:fl-secagg-quant}.
Compared with the process shown in Algorithm \ref{algo:sec-fl-algo}, the main difference is that in each round, after training a local model, each selected client sends an obscured quantized model update to the cloud server.
The cloud server then performs aggregation over these obscured quantized model updates and produces an updated global model.

Note that following prior works \cite{KonecnM16,ReisizadehMHJP20}, the quantization in our system is performed over the difference between a locally trained model $\mathbf{w}^t_k$ and the current global model $\mathbf{w}^{t-1}$, i.e., $\Delta^t_k=\mathbf{w}^t_k-\mathbf{w}^{t-1}$, as the model difference is more amenable to compression, in contrast to the locally trained model.
Note that the quantizer requires the input to be bounded in a range $[-B,B]$. This can be achieved by clipping the values in $\Delta^t_k$ based on the threshold $B$ which can be pre-set based on a public calibration dataset \cite{ShokriS15}.
Values greater than $B$ are set to $B$, and to $-B$ if they are smaller than $B$. 
In what follows, we further address some practical considerations.

\noindent\textbf{Preventing overflow in aggregating quantized values.} As the aggregation is performed over the quantized values from multiple clients, preventing overflow is crucially important. We note that this can be achieved via applying a scaling mechanism on the input range for quantization \cite{zhang2020batchcrypt}.
In particular, to prevent overflow when the aggregation is performed over $c$ clients, we can set the input range for quantization as $[-c\cdot B, c\cdot B]$, in contrast with the original range $[-B, B]$.
The intuition here is that by scaling the input range, each quantized value is scaled down $c$ times so overflow can be prevented when $c$ quantized values are aggregated.

\noindent \textbf{Identifying negative values in de-quantization.} It is noted that as the aggregate model update $\Delta^t$ is derived modulo $2^r$, the cloud server needs to know which values in the $\Delta^t$ should have been negative indeed before performing the de-quantization.
This can be achieved as follows.
The cloud server checks if a value $a$ in $\Delta^t$ is greater than $2^{r-1}-1$.
if yes, its raw value can be obtained via such transformation $a'=a-2^{r}$.
Otherwise, it should have been non-negative.
After this check and the transformation applied for values deemed to be negative, the cloud server can then proceed to the de-quantization.

\section{Hardening Secure Federated Learning}
\label{sec:integrity-assued-design}

In our design above, we assume a passively adversarial cloud server that faithfully performs the designated computation.
We now present a practical design which can provide security against an actively adversarial cloud server with minimal performance overhead (as demonstrated in the experiments later), ensuring the correctness of computation at the cloud server.
We aim for a pragmatic solution and thus following the trend of using trusted hardware, with only minimal assumption on its computation integrity assurance.

\begin{figure}[!t]
\centering

\fbox{
  \begin{minipage} [t]{0.46\textwidth}



\begin{enumerate}[1:]
\item Each client $\mathcal{C}_i$ generates the public key $g^{x_i}$ and signature $\sigma_{\mathcal{C}_i}=\mathsf{Sign}_{sk_{\mathcal{C}_i}}(g^{x_i})$.

\item  Each client $\mathcal{C}_i$ sends to $g^{x_i}$ and $\sigma_{\mathcal{C}_i}$ to $\mathcal{F}_{\mathsf{TEE}}$ as well as revealed to the cloud server.

\item The cloud server invokes $\mathcal{F}_{\msf{TEE}}$ on (``$\msf{Compute}$'', ($\{g^{x_i}\}_{i\in \mathcal{S}}$, $\{\sigma_{\mathcal{C}_i}\}_{i\in \mathcal{S}}$)), and receives $\perp$ or updated state containing $(\mathsf{ctr}, \{g^{x_j}\}_{j\in \mathcal{S}, j\neq i},\sigma)$ to be sent to each client $\mathcal{C}_i$, where $\sf{ctr}=0$.

\item Each client $\mathcal{C}_i$ runs $\mathsf{Verify}_{vk_{\msf{T}}}((\mathsf{ctr},\{g^{x_j}\}_{j\in \mathcal{S}, j\neq i}), \sigma)$, upon receiving $(\msf{ctr}, \{g^{x_j}\}_{j\in \mathcal{S}, j\neq i},\sigma)$.
\item Each client continues to produce the keys used for blinding factor generation if the check passes, or aborts otherwise.

\end{enumerate}
\end{minipage}
}
\caption{Key setup in security-hardened federated learning in our system.}
\label{fig:key-setup-tee-FL}
\end{figure}

\subsection{Abstract Functionality from Trusted Hardware}

Inspired by the prior work \cite{TramerZLHJS17,Kumar2020}, we define the abstract functionality $\mathcal{F}_{\sf{TEE}}$ assumed out of the trusted hardware, which we use to harden the secure federated learning service in our system.
The functionality is parameterized by a secure signing scheme $\{\mathsf{Sign}_{sk},\mathsf{Verify}_{vk}\}$ with a key pair $(sk,vk)$.
Let $\mathsf{Sign}_{sk}(m)$ denote signing a message $m$ and and $\mathsf{Verify}_{vk}(\sigma, m)$ denote verifying a signature $\sigma$ on $m$.
The functionality allows users to provide a program $\msf{prog}$ using the ``$\sf{install}$'' command.
It then returns $\alpha=\msf{Sign}_{sk}(\mathsf{Hash}(\sf{prog}))$ as a token for this program, which allows public integrity verification given the program $\sf{prog}$ and the signature verification key $vk$.
Subsequent invocation of $\mathcal{F}_{\sf{TEE}}$ runs $\sf{prog}$ on given inputs $\sf{inp}$ and fresh randomness $\sf{rnd}$ using the ``$\sf{Compute}$'' command, and produces some output $\sf{outp}$.
The program $\sf{prog}$ could be stateful and may receive successive inputs in different rounds as indexed by a counter $\sf{ctr}$ and produce corresponding output $\sf{outp}_{\sf{ctr}}$.

Let $\sf{state}_{\sf{ctr}}$ be an internal state maintained by $\mathcal{F}_{\sf{TEE}}$.
Upon the initialization, the $\sf{state}_{\sf{0}}$ is empty (i.e., $\sf{state}_{\sf{0}}=\epsilon$).
On each ``$\sf{Compute}$'' command for the $\sf{ctr}$-th round, the functionality produces the output $\sf{outp}_{\sf{ctr}}$ and a signature $\sigma_{\sf{ctr}}$ on $(\sf{outp}_{\sf{ctr}},\sf{ctr})$.
Then, $\sf{state}_{ctr-1}$ is updated to $\sf{state}_{ctr}$ by adding the tuple $\{\mathsf{ctr},\sf{inp}_{\sf{ctr}}, \sf{outp}_{\sf{ctr}}, \sigma_{\sf{ctr}},\sf{rnd}_{\sf{ctr}} \}$.
%
%
The updated state $\sf{state}_{ctr}$ is always given to the host of the trusted hardware. This reflects the assumption that the internal state of the trusted hardware can be observed by the host.
Note that for the output from the functionality, any parties with the public $vk$ can verify the integrity. 
In short, we only assume that the functionality can conduct computation and produce signed outputs.

\subsection{Security-Hardened Federated Learning}

Given the functionality $\mathcal{F}_{\sf{TEE}}$, we now describe how to make our above federated learning service secure against with the cloud server that may not correctly conduct the designated processing.
The main idea is to have every message sent by the (TEE-enabled) cloud server in the semi-honest protocol be computed by $\mathcal{F}_{\sf{TEE}}$, where the signing key pair $(sk_{\sf{T}},vk_{\sf{T}})$ is used by $\mathcal{F}_{\sf{TEE}}$.
These messages can be verified by any party which knows $vk_{\sf{T}}$.
We assume that every participant in the system knows $vk_{\sf{T}}$ in a reliable manner.
Later we will show how this can be achieved via the remote attestation mechanism in the widely popular trusted hardware Intel SGX.
We also assume that each client has a signing key pair $(sk_{\mathcal{C}_i}, pk_{\mathcal{C}_i})$, where the pubic key $pk_{\mathcal{C}_i}$ is bound to each client's identity $\mathcal{C}_i$.

We now describe our protocol for security-hardened federated learning.
Firstly, the requester sends to the cloud server the code $\mathsf{prog}$.
%
The cloud server invokes $\mathcal{F}_{\mathsf{TEE}}$ on (``$\mathsf{Install}$ '', $\mathsf{prog}$) to receive $(\mathsf{state}_0, \alpha=\mathsf{Sign}_{sk_{\msf{T}}}(\mathsf{Hash}(\msf{prog})))$, which is also sent to the requester and all clients in the system.
The requester and each client run $\msf{Verify}_{vk_{\msf{T}}}(\alpha, \mathsf{Hash}(\msf{prog}))$, and abort if the check fails.
After $\mathcal{F}_{\mathsf{TEE}}$ is initialized, the key setup for secure aggregation is conducted as shown in Fig. \ref{fig:key-setup-tee-FL}.
After the key setup is done, secure federated learning now proceeds as shown in Fig. \ref{fig:secure-agg-TEE}.

\begin{figure}[!t]
\centering

\fbox{
  \begin{minipage} [t]{0.46\textwidth}

\begin{enumerate}[1:]
\item In the beginning of the round $\sf{ctr}$, the cloud server invokes $\mathcal{F}_{\msf{TEE}}$ on (``$\msf{Compute}$'', $\msf{ctr}$) and receives $(\msf{ctr}, \mathcal{T}_{\msf{ctr}}, \mathbf{w}^{\msf{ctr}-1}, \sigma)$ as the updated state, which is also sent to each client $\mathcal{C}_k$ in $\mathcal{T}_{\msf{ctr}}$.

\item Each selected client in $\mathcal{T}_{\msf{ctr}}$ runs $\sf{Verify}_{vk_{\sf{T}}}((\msf{ctr}, \mathcal{T}_{\msf{ctr}}, \mathbf{w}_{\msf{ctr}-1}), \sigma)$, and aborts if the verification fails.

\item Each selected client $\mathcal{C}_k$ does local training and produces an obscured (quantized) model update $\mathbf{w'}_k$. 

\item The ($\msf{ctr}$, $\mathcal{C}_k$,  $\mathbf{w'}_k$, $\sigma_{\mathcal{C}_k}$) is sent to $\mathcal{F}_{\msf{TEE}}$ as well as revealed to the cloud server.

\item The cloud server invokes $\mathcal{F}_{\msf{TEE}}$ on (``$\msf{Compute}$'', $\msf{ctr}$, $\{\mathcal{C}_k, \mathbf{w}'_k, \sigma_{\mathcal{C}_k}\}_{k\in \mathcal{T}_{\msf{ctr}}})$.
It receives $(\msf{ctr+1}, \mathcal{T}_{\msf{ctr+1}}, \mathbf{w}^{\msf{ctr}}, \sigma)$ for the next round, or $\perp$ if verification any selected clients' signatures fails.
\end{enumerate}
\end{minipage}
}
\caption{Aggregation in our security-hardened federated learning design.}
\label{fig:secure-agg-TEE}
\end{figure}

\begin{figure*}[!t]
\centering
  \begin{minipage}[t]{0.31\linewidth}
    \centering{\includegraphics[width=\linewidth]{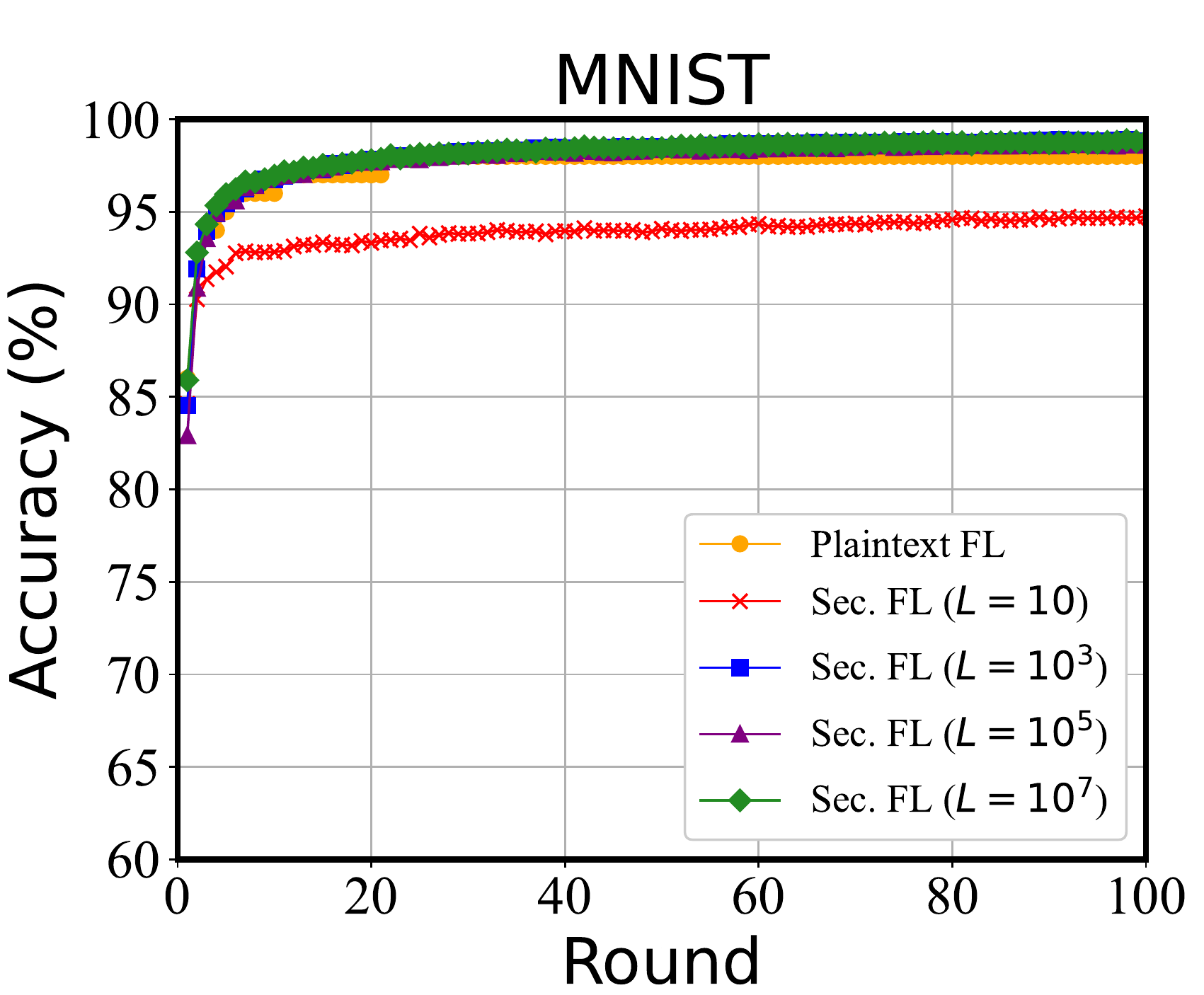}}
  
  \end{minipage}%
  \begin{minipage}[t]{0.31\linewidth}
    \centering{\includegraphics[width=\linewidth]{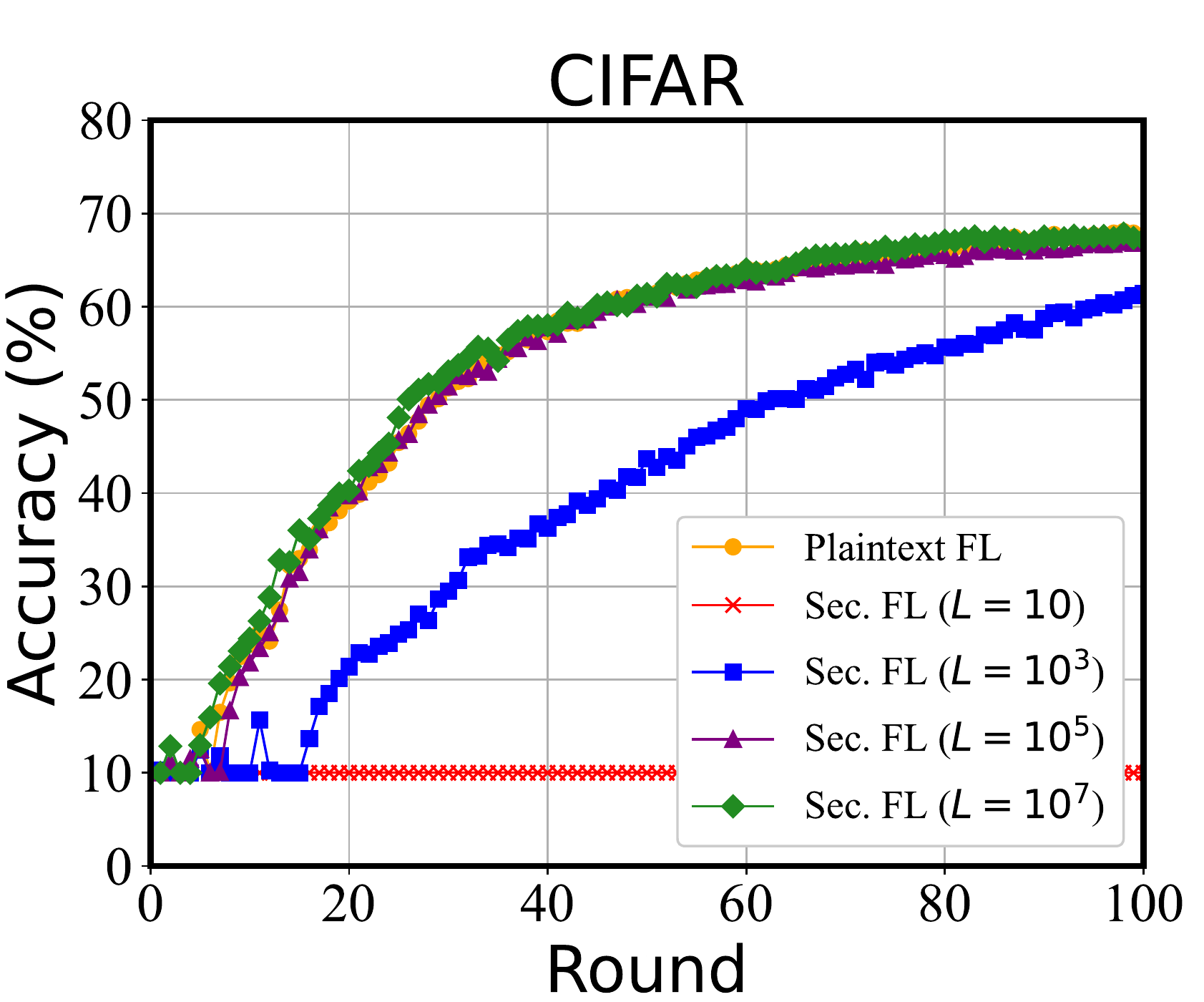}}
        
  \end{minipage}
  \begin{minipage}[t]{0.31\linewidth}
    \centering{\includegraphics[width=\linewidth]{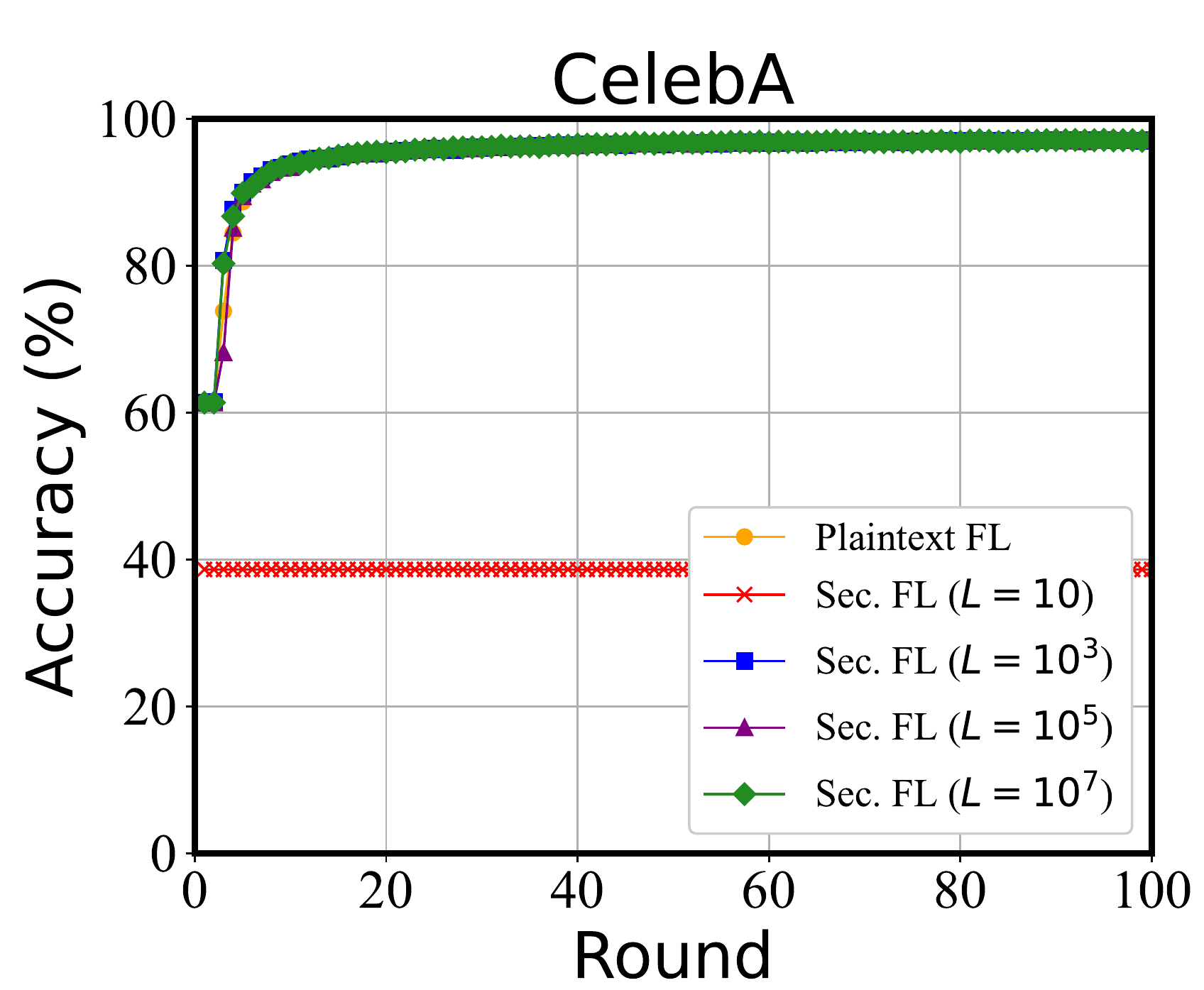}}
  \end{minipage}
  \caption{Effect of the scaling factor $L$ on accuracy in the proposed secure federated learning design over different datasets.}
  \label{fig:effect_scaling_factor}
\end{figure*}

\subsection{Realization of the Ideal Functionality $\mathcal{F}_{\sf{TEE}}$}

The above functionality $\mathcal{F}_{\sf{TEE}}$ can potentially be realized via any trusted hardware techniques that provide code attestation and signing.
As an instantiation, we resort to Intel SGX due to its wide integration in commodity processors.
SGX enables the execution of programs in secure enclaves that are isolated from all other applications on the same host.
It also provides attestation mechanisms which assure the integrity of the program loaded into the enclave once it has been attested. 
In particular, a signed attestation report can be generated for the program being loaded, which allows anyone to verify based on a public key corresponding to Intel's report signing key.
The signing key pair $(sk_{\sf{T}},vk_{\sf{T}})$ of the functionality $\mathcal{F}_{\sf{TEE}}$ can be generated inside the enclave, and the verification key $vk_{\sf{T}}$ can be part of the payload of the signed attestation report so that it is made public reliably.
This verification key $vk_{\sf{T}}$ can then be used to verify the signed outputs from the functionality throughout the computation procedure in federated learning.

\section{Experiments}
\label{sec:experiments}

\subsection{Setup}

%
%
%

We implement the client with Python.
The PyTorch is used for model training.
%
%
%
%
For the federated learning service, we use C++ since the official SGX SDK only has C++ interfaces.
We adopt Thrift~\cite{thrift} to enable cross-language communication.
%
That is, the federated learning service is implemented as a daemon on a Thrift server that provides the corresponding interface which can be invoked by a service client generated as Python scripts.
Each client uses the service client to connect with the server to upload the model update.
%
%
The daemon passes the model update to the enclave and gets the global model and its signature as the result.
In our experiments, we manage to get rid of expensive EPC paging in the enclave by temporarily storing  all received (obscured) model updates and their signatures in the untrusted memory of the host server.
Aggregation in the enclave is then performed by reading in and verifying the (obscured) model updates one by one.

We use three popular datasets in our experiments, including the commonly used datasets MNIST and CIFAR-10, and the CelebA dataset from the LEAF federated learning benchmark framework \cite{Caldas19}.
The MNIST dataset contains images of 0-9 handwritten digits, with 60000 training examples and 10000 testing examples. 
The CIFAR-10 dataset contains 50000 training color images and 10000 testing color images, with 10 classes.
The CelebA dataset contains 200,288 celebrity images, each with 40 attribute annotations.
In this paper we use this dataset for the application of gender classification.
For MNIST, we use a relatively simple CNN model with two 5x5 convolution layers, a dropout layer, and two fully connected layers ($21,840$ total parameters).
For CIFAR-10, we rely on the popular AlexNet \cite{KrizhevskySH12} with more sophisticated structures and larger sizes but use less conv. kernels ($23,272,266$ total parameters).
For CelebA, we use ResNet-18 \cite{HeZRS16}, a popular CNN model that is 18 layers deep ($13,962,562$ total parameters).
Hereafter, to facilitate presentation, we will refer to the models over different datasets as the MNIST model, CIFAR model, and CelebA model, respectively.
All evaluations are conducted on an SGX-enabled server equipped with Intel Xeon E 2288G 3.70GHz CPU (8 cores), 128GB RAM, and 3 RTX 3090 GPUs.

\begin{figure*}[!t]
\centering
\small
  \begin{minipage}[t]{0.3\linewidth}
    \centering{\includegraphics[width=\linewidth]{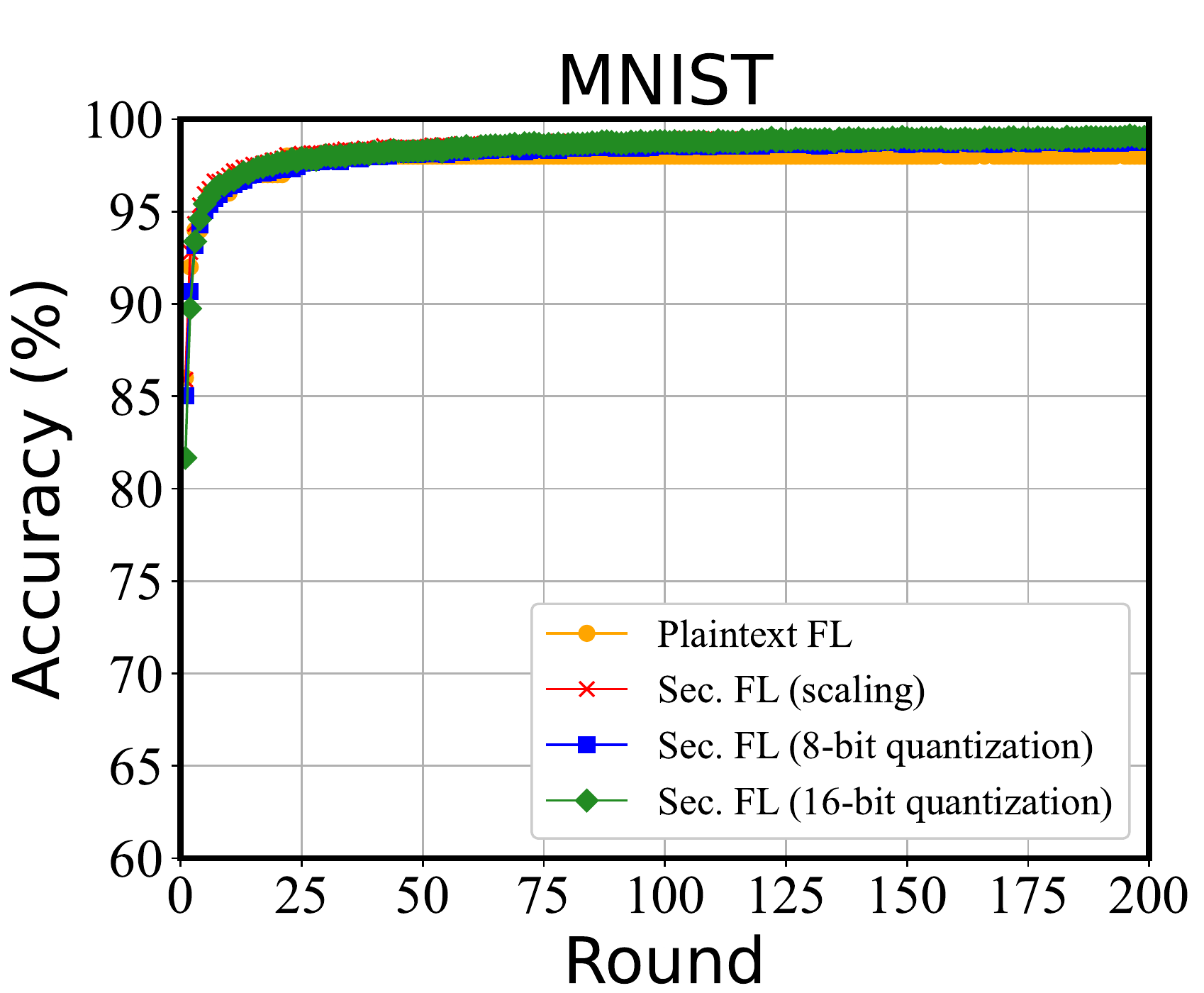}}

  \end{minipage}%
  \begin{minipage}[t]{0.3\linewidth}
    \centering{\includegraphics[width=\linewidth]{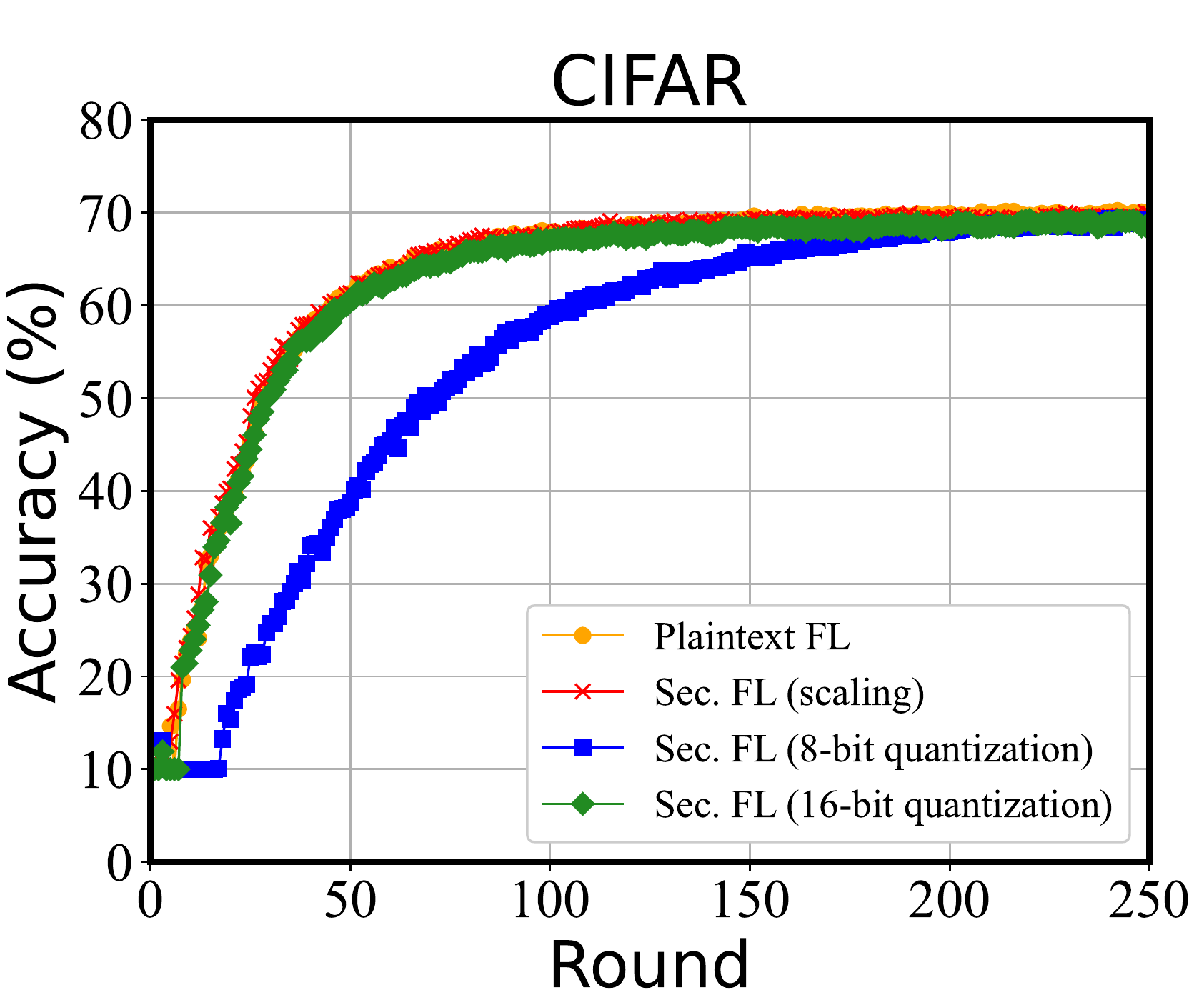}}
  \end{minipage}
    \begin{minipage}[t]{0.3\linewidth}
    \centering{\includegraphics[width=\linewidth]{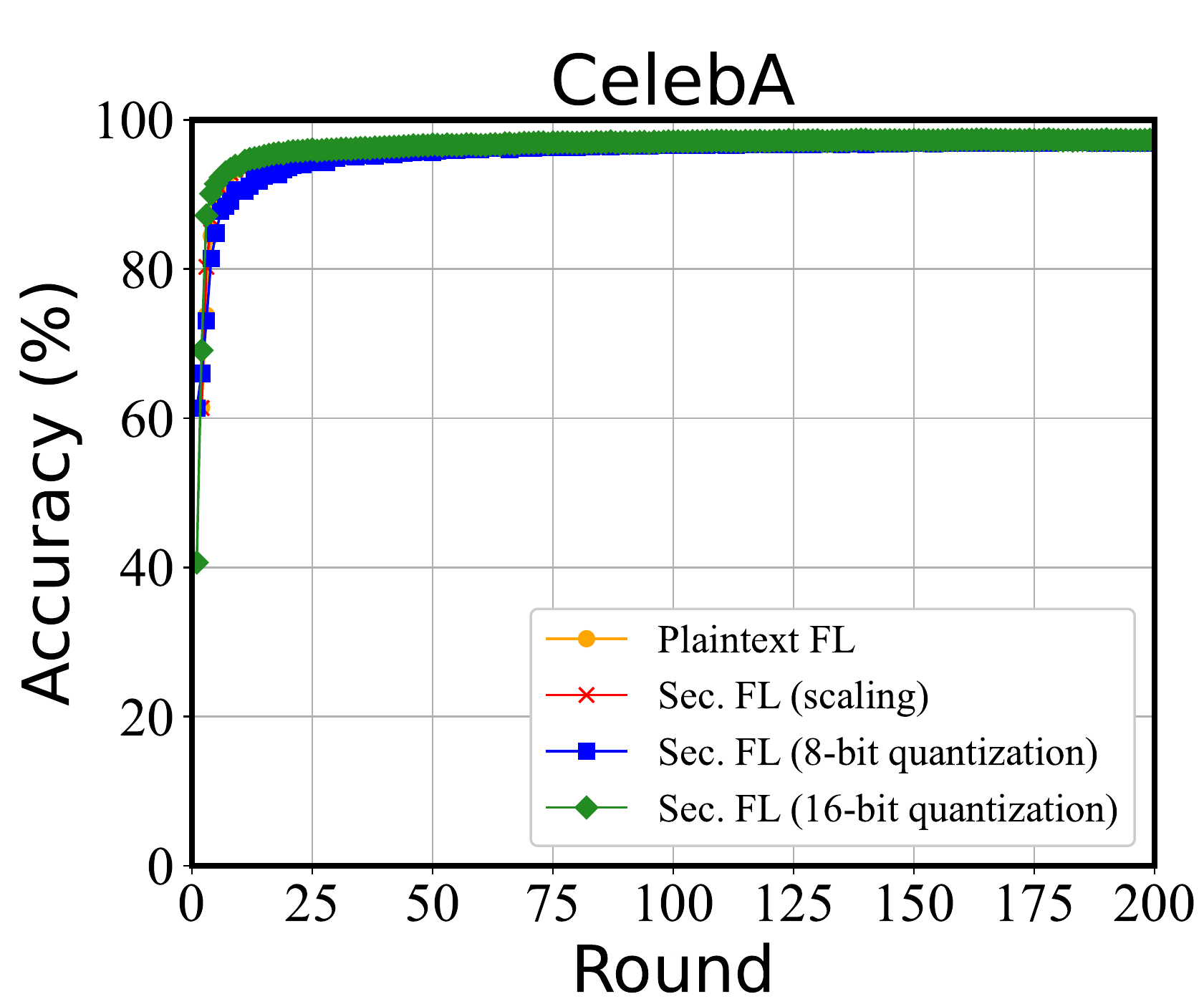}}
  \end{minipage}
             \caption{Accuracy evolution of the models over different datasets.}
    \label{fig:accuracy_evolution}
\end{figure*}

  \begin{figure}[t]
    \centering{\includegraphics[width=0.6\linewidth]{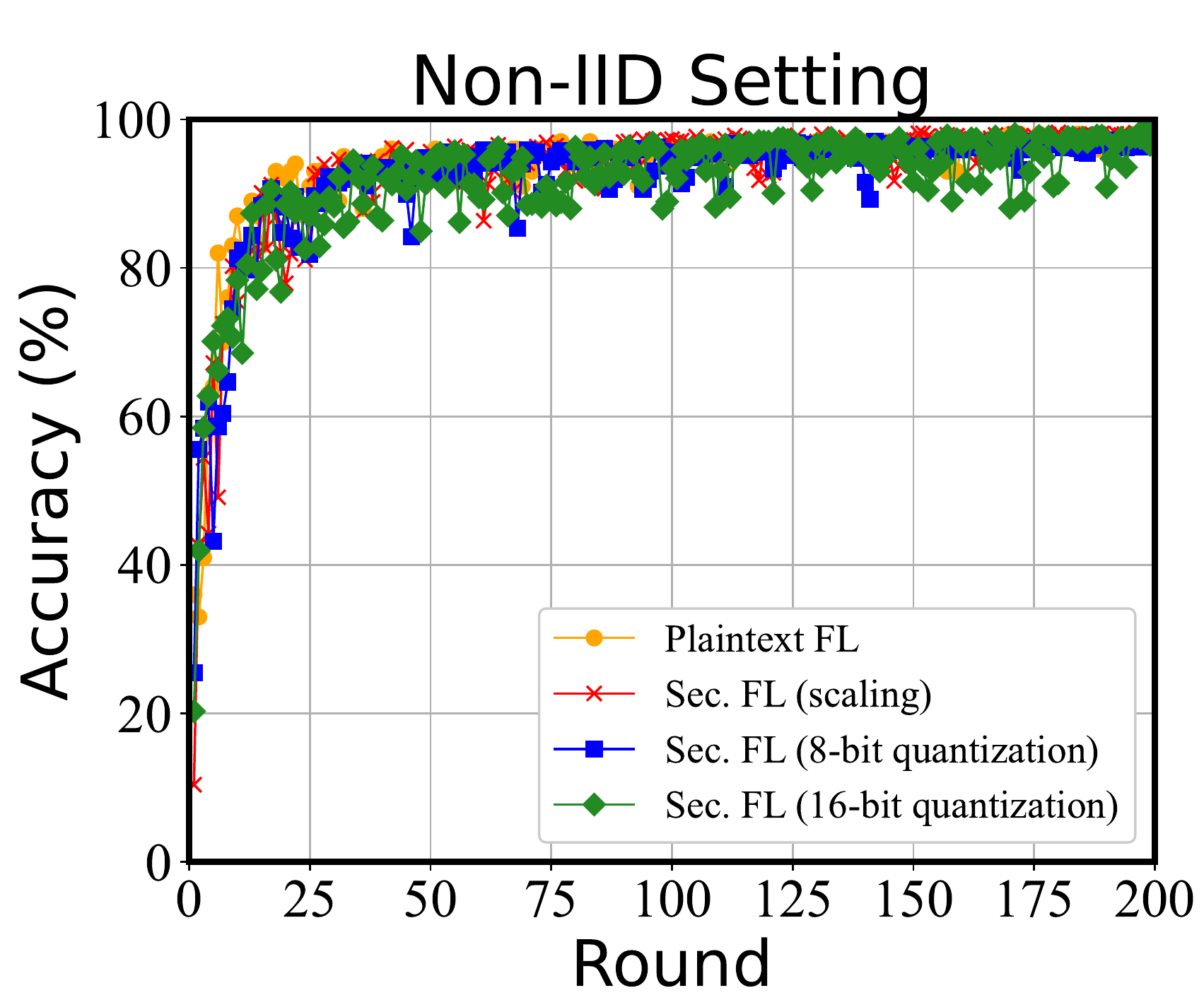}}
    \caption{MNIST model accuracy evolution under the non-IID setting.}
    \label{fig:accuracy_result_noniid}
  \end{figure}

\begin{figure*}[!t]
\centering
\small
  \begin{minipage}[t]{0.3\linewidth}
    \centering{\includegraphics[width=\linewidth]{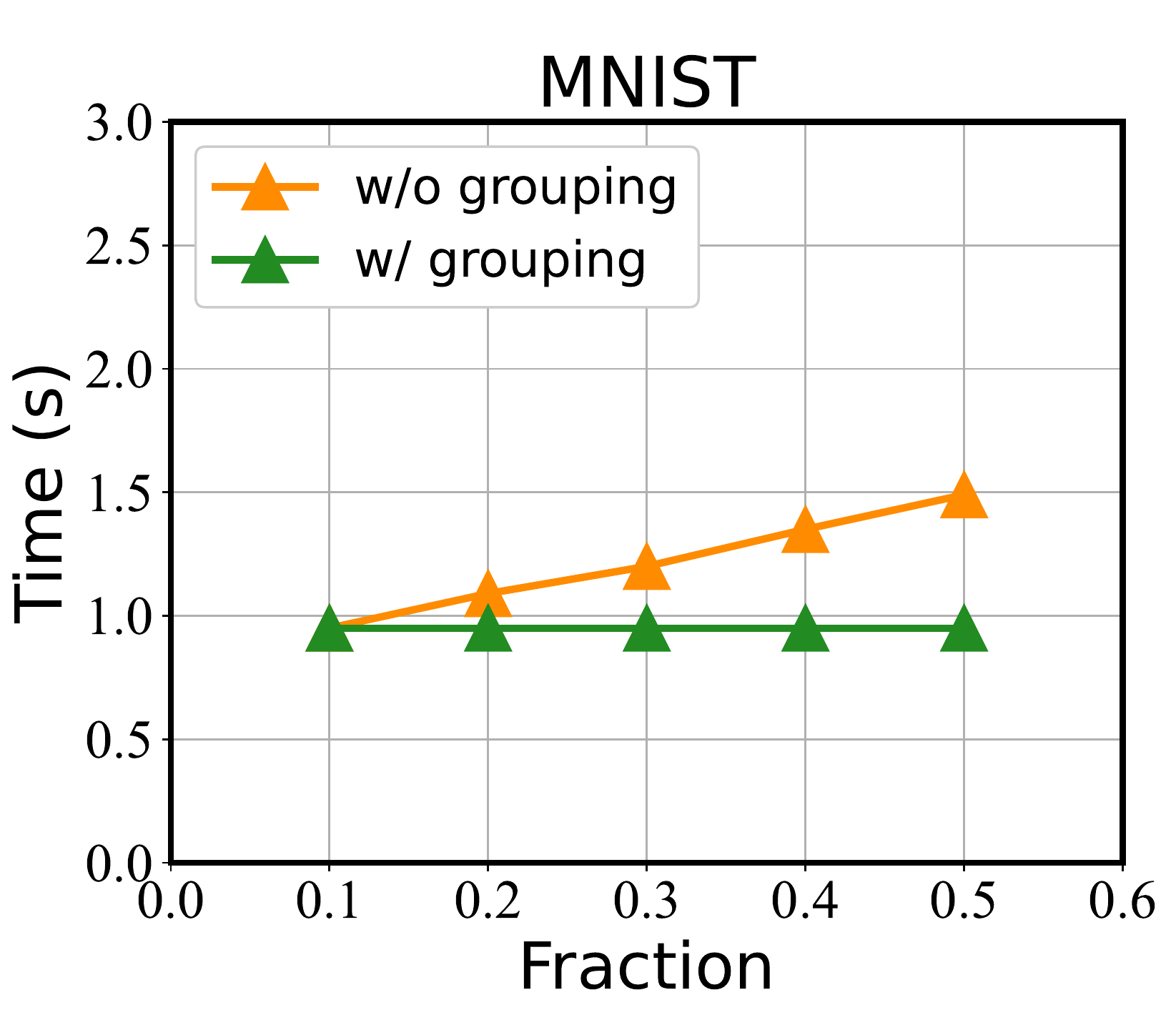}}
    
  \end{minipage}%
  \begin{minipage}[t]{0.3\linewidth}
    \centering{\includegraphics[width=\linewidth]{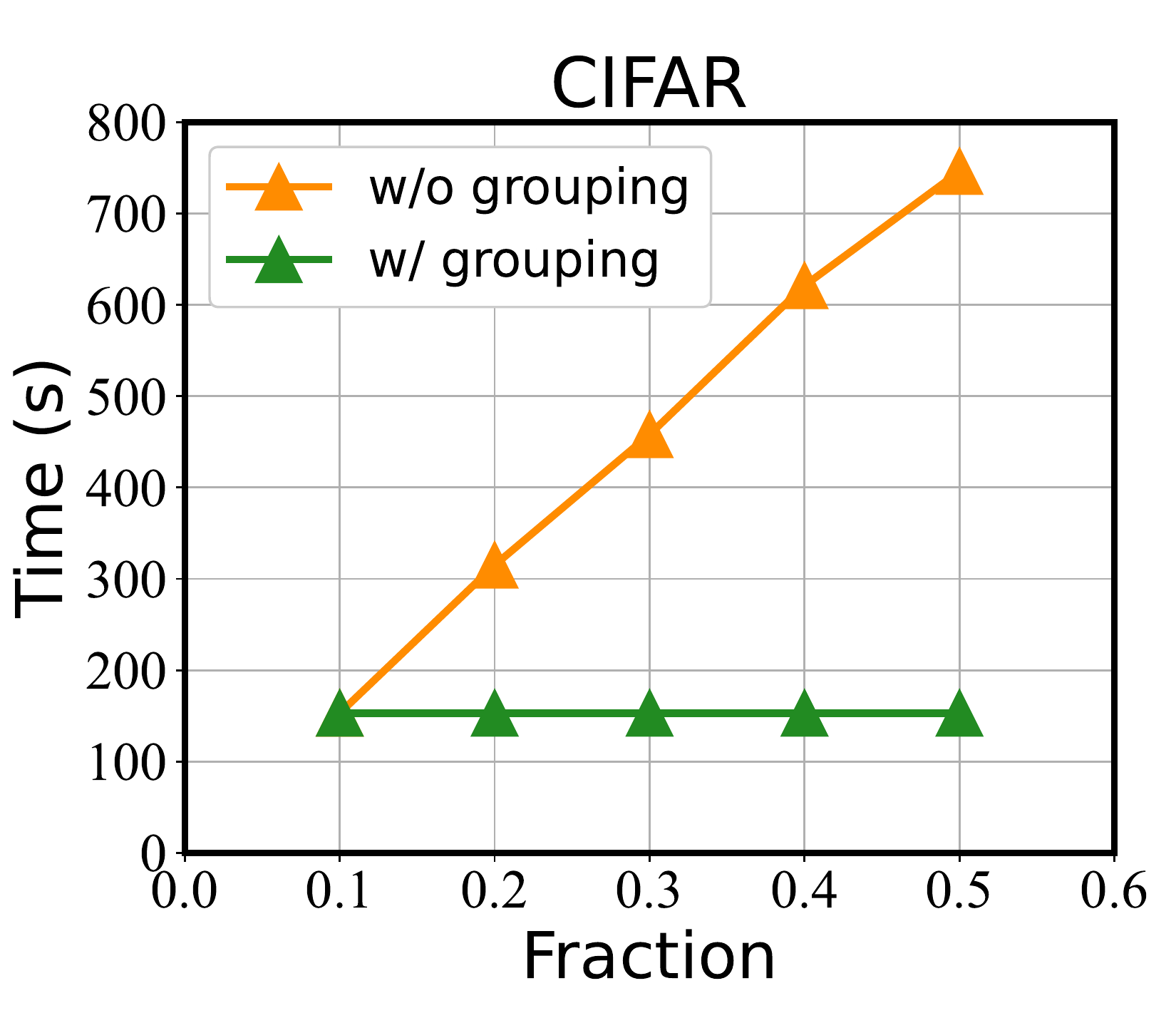}}

  \end{minipage}
    \begin{minipage}[t]{0.3\linewidth}
    \centering{\includegraphics[width=\linewidth]{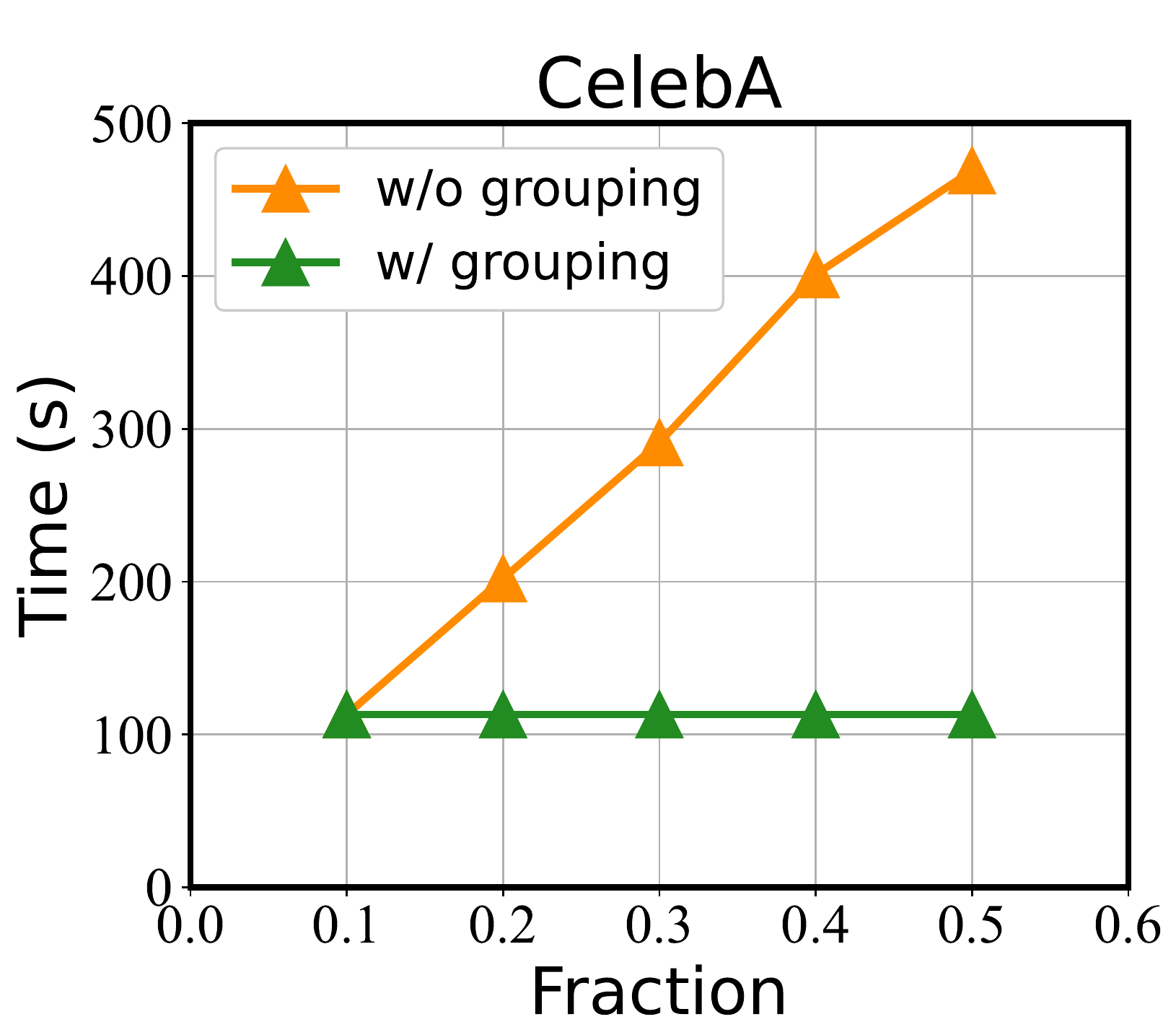}}
  \end{minipage}
  \caption{Client's encryption cost with varying fraction of selected clients per round, over different models.}
    \label{fig:client_encryption_cost}
\end{figure*}

\subsection{Accuracy Evaluation}
We evaluate three cases: the plaintext federated learning with no privacy of model updates, our basic secure protocol without quantization, where a scaling factor is used to scale model parameters into integers for cryptographic computation, and our secure protocol extended with quantization. 

For the datasets MNIST and CIFAR, we randomly shuffle the training examples and evenly distribute them across $100$ clients. 
Each client receives $600$ examples under the MNIST dataset, and $500$ examples under the CIFAR-10 dataset respectively.
For CelebA, we divide it into a training set and a test set with a 70/30 ratio, and distribute the training set to the clients, and each client holds about 2002 images.
Such way of partitioning is referred to as IID data distribution \cite{McMahanMRHA17}.
The fraction of clients being selected in each round is set to $10\%$.
Each selected client performs local training over $5$ epochs, with the learning rate being $0.01$ and batch size being $10$.
The threshold $B$ related with quantization is set to $0.5$ for the MNIST model and $0.1$ for the CIFAR model and the CelebA model, respectively.
%
%
For our secure protocol extended with quantization, we examine the cases where the quantization bit widths are $8$ and $16$ respectively.

To start, we evaluate the effect of varying scaling factors on the model accuracy in our secure federated learning design, of which the results are shown in Fig. \ref{fig:effect_scaling_factor}.
It is observed that the use of an appropriate scaling factor does not adversely affect the accuracy as the number of rounds grows. As long as a large scaling factor is used, the model accuracy is maintained with respect to the plaintext baseline. Such accuracy preservation is also consistent with the literature that uses the trick of scaling factor for cryptographic computation (e.g., \cite{DBLP:journals/tdsc/ZhengDY020,LiuJLA17,ZhengDW18}, to just list a few).
Hereafter, we set the scaling factor to $10^7$ in all the remaining experiments.

In Fig. \ref{fig:accuracy_evolution}, we show the evolution of the testing accuracy under varying number of rounds over the MNIST dataset, CIFAR dataset, and CelebA dataset, respectively.
It is noted that the legend ``Plaintext FL'' refers to the plaintext federated learning setting where the raw values of model updates in 32-bit floating-point representation are used.
The legend ``scaling'' refers to the setting where a scaling factor ($10^7$) is used to scale up the fractional values in model updates into 32-bit integers to support computation in the cryptographic aggregation protocol.
The legend ``8/16-bit quantization'' refers to the setting where a 8-/16-bit quantizer is applied over the fractional values in model updates, leading to 8-/16-bit integers for supporting communication efficiency optimization in the ciphertext domain.
It can be observed that our secure protocols share similar behavior with the baseline and achieve comparable accuracy.
The unique integration of quantization (with adequate bit width) does not adversely affect the quality of the trained model either.

In addition to the IID setting, following the seminal work \cite{McMahanMRHA17} on federated learning, we further examine the non-IID setting over the MNIST dataset.
The non-IID data distribution is set up by sorting the training images according to the digit labels, dividing them into shards of size $300$, and randomly distributing two shards to each client.
%
%
We show the accuracy evaluation result in Fig. \ref{fig:accuracy_result_noniid}, which demonstrates similar behavior and comparable accuracy to the plaintext baseline.
Note that data heterogeneity is related to the federated learning paradigm itself and is independent of our security design.
There are no ties between data heterogeneity and client drop-out.
Client-dropout only affects the actual number of model updates being aggregated to produce the global model. Our system ensures the correctness of secure aggregation even in case of client-dropout.

\begin{figure*}[!t]
\centering
\small
  \begin{minipage}[t]{0.3\linewidth}
    \centering{\includegraphics[width=\linewidth]{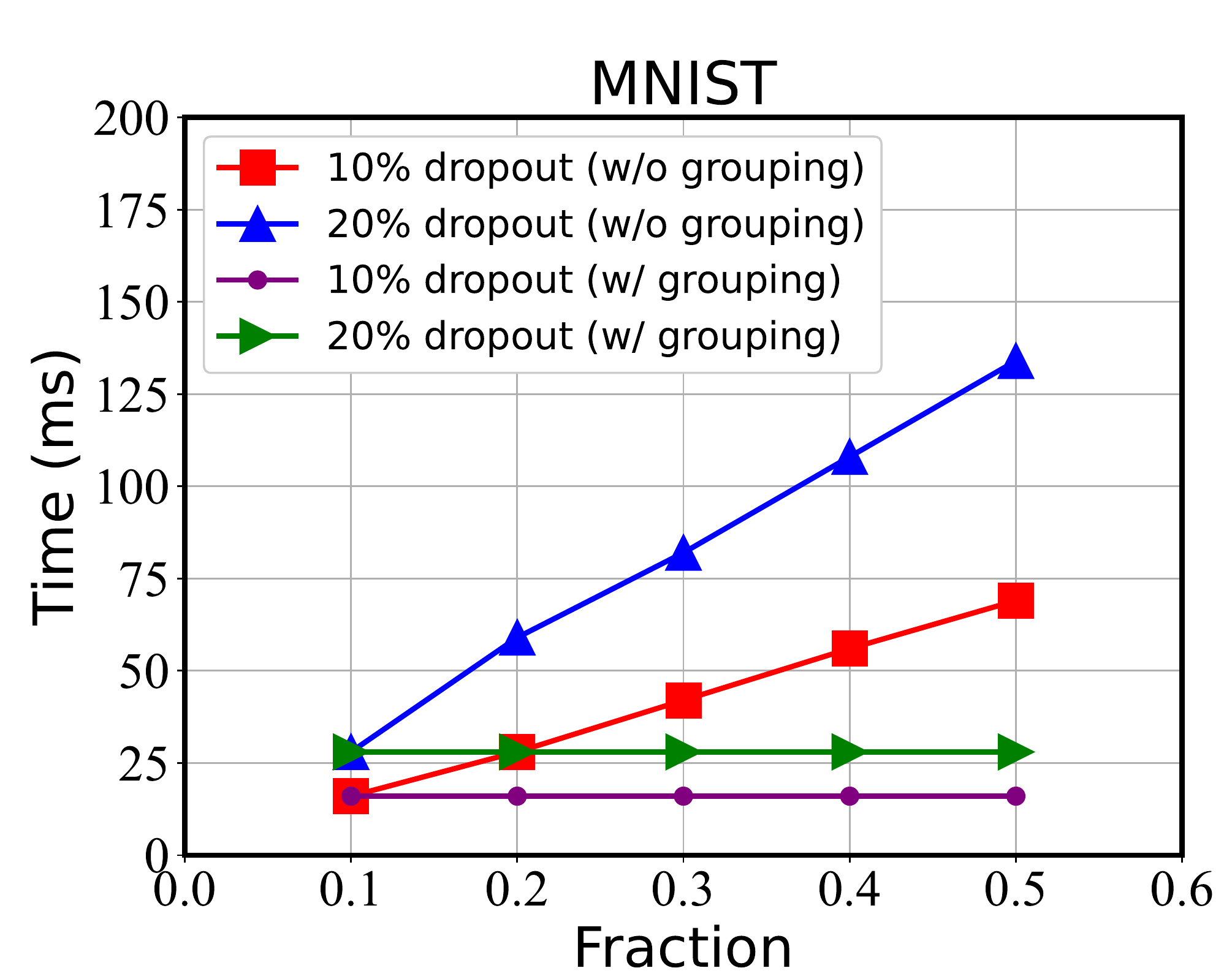}}
    
  \end{minipage}%
  \begin{minipage}[t]{0.3\linewidth}
    \centering{\includegraphics[width=\linewidth]{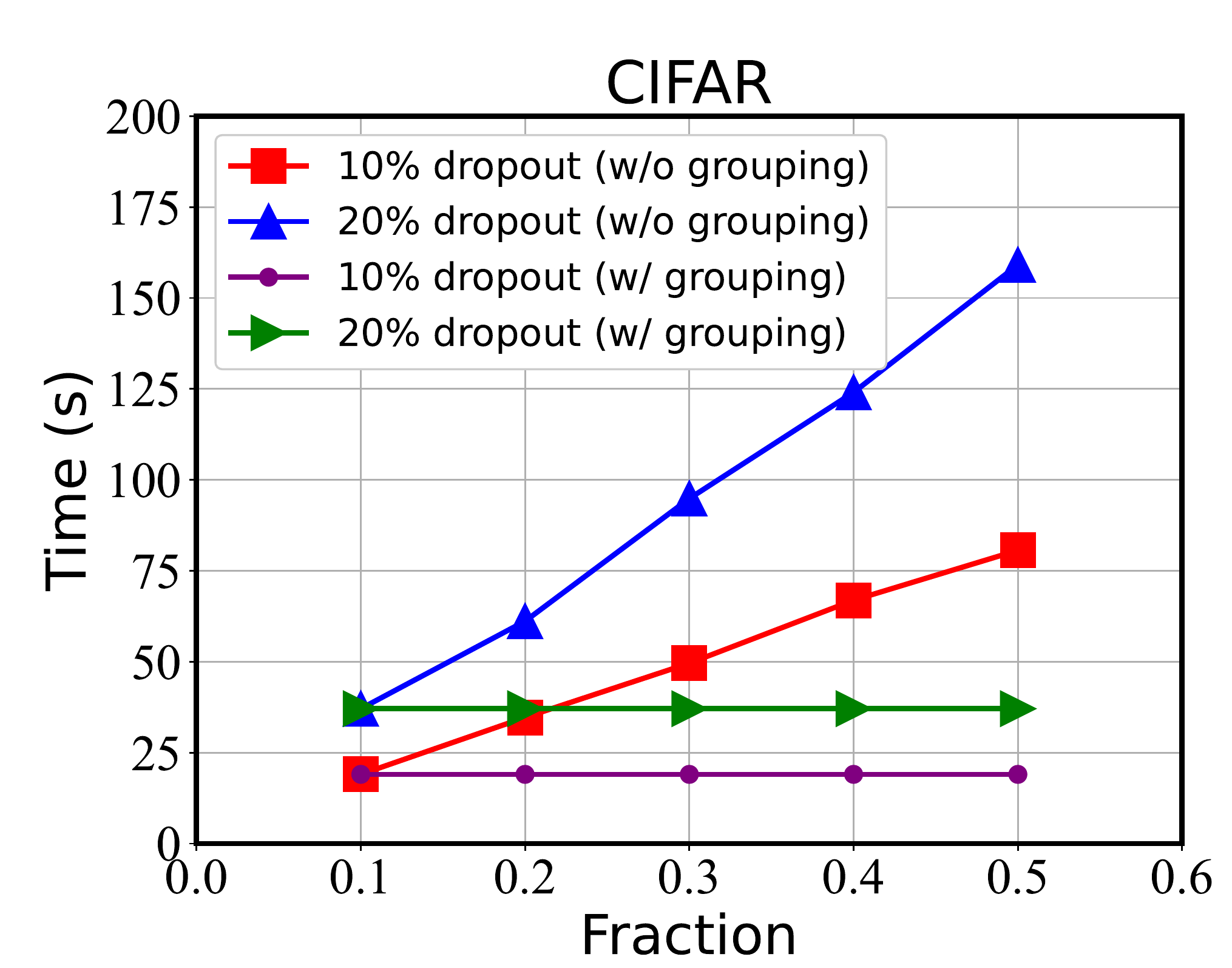}}

  \end{minipage}
    \begin{minipage}[t]{0.3\linewidth}
    \centering{\includegraphics[width=\linewidth]{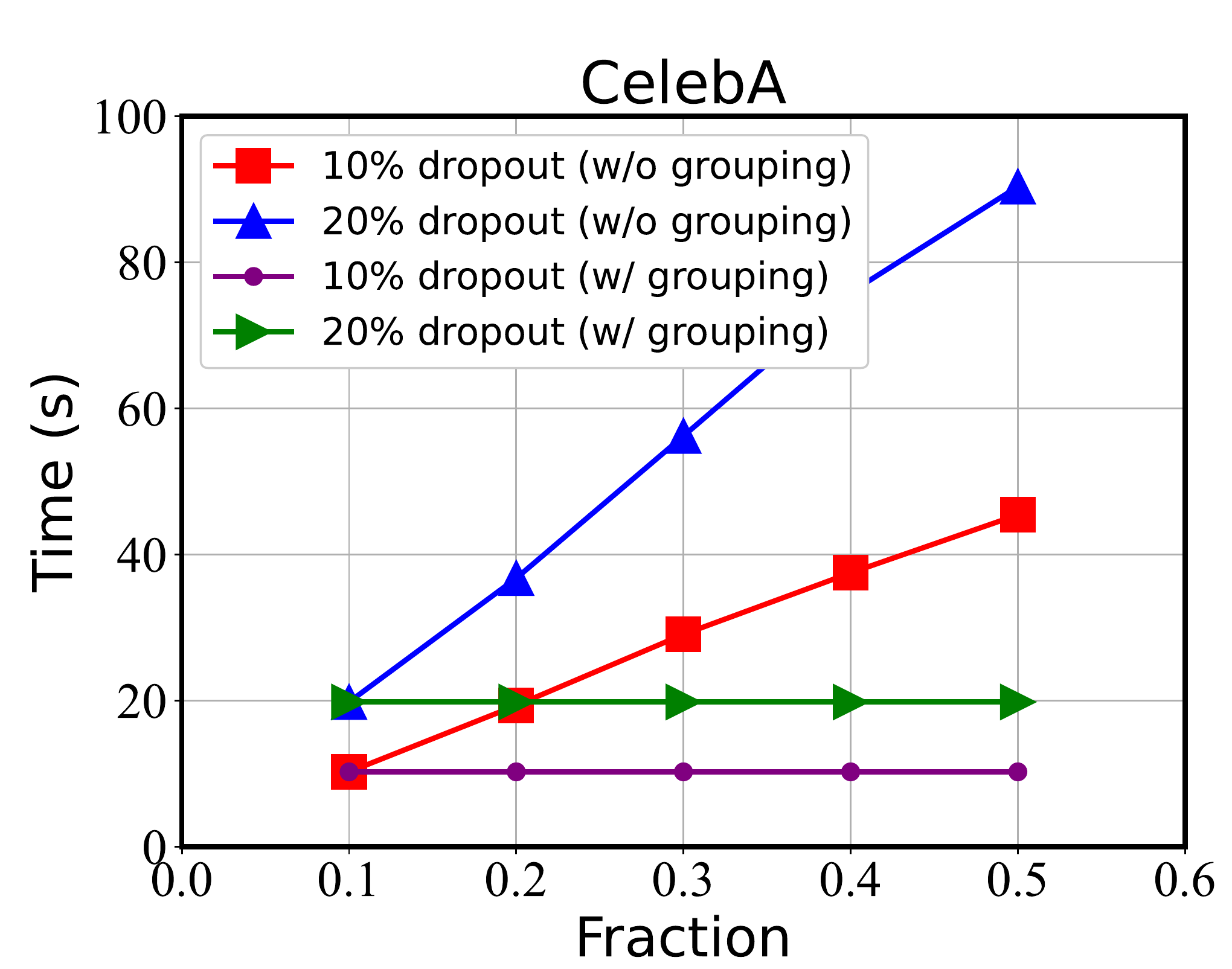}}
  \end{minipage}
  \caption{Client's computation cost in dealing with different dropout rates, over different models.}
    \label{fig:client_recovery}
\end{figure*}

\begin{table}[t!]
\caption{Size of Model Updates (in MB)}
\centering
\begin{tabular}{@{}lccc@{}}
\toprule
Setting                       & MNIST & CIFAR & CelebA \\ \midrule
Plaintext FL                  & 0.083 & 88.78 &   53.26     \\
Sec. FL (scaling)             & 0.083 & 88.78 &   53.26     \\
Sec. FL (16-bit quantization) & 0.042 & 44.49 &    26.63    \\
Sec. FL (8-bit quantization)  & 0.021 & 22.19 &    13.32    \\ \bottomrule
\end{tabular}
\label{tbl:model_size}
\end{table}

\subsection{Performance Evaluation}

\subsubsection{Client-Side Performance}
We first examine the client-side computation performance.
Recall that encrypting a model update requires generation of blinding factors that requires $O(|\mathcal{T}|)$ hashing operations per each, where $|\mathcal{T}|$ is the number of selected clients in a round.
The encryption cost thus scales with $|\mathcal{T}|$ (and inherently the size of the model update).
We show the client's running time of encrypting (quantized) model updates for varying fraction of clients being selected in a round, over different models in Fig. \ref{fig:client_encryption_cost}.
%
%
%
For the smallest MNIST model, it is seen that the encryption cost is on the order of a few seconds.
For the CIFAR and CelebA models, the running times are on the order of minutes due to their substantially larger model size ($23,272,266$ parameters and $13,962,562$ parameters as opposed to $21,840$ parameters of the MNIST model).
However, it is worth noting that our system can flexibly support client grouping to largely limit the computation complexity of a client (i.e., independent of the fraction), as demonstrated in Fig. \ref{fig:client_encryption_cost} where the selected clients are grouped with size $10$. 
%

We also examine the computation cost of a client in a recovery phase to handle dropouts, with the results are shown in Fig. \ref{fig:client_recovery}, under varying dropout rates over different models.
As expected, the running time of the client scales linearly with the dropout rate.
It is also revealed that client grouping can greatly reduce the cost.
%

In Table \ref{tbl:model_size}, we report the size of a model update (i.e., sizes of model update parameters), for the plaintext case and our secure design under the settings of scaling, 8-bit quantization, and 16-bit quantization, respectively.
Our secure protocol with scaling incurs no overhead on the model update size, as the bit precision remains the same. 
Our secure protocol can achieve $4\times$ reduction under 8-bit quantization and $2\times$ reduction under 16-bit quantization on the size of transferred model update in a round.
%
%
Recall that our system exhibits similar behavior in accuracy evolution with regard to varying number of rounds.
Given a target number of rounds, our system with quantization can lead to $2\times$ or $4\times$ reduction on the communication, with comparable accuracy to the plaintext baseline.

\begin{table}[t!]
\caption{Performance Complexity Comparison}
\label{tbl:complexity}
\centering
\begin{tabular}{@{}ccc@{}}
\toprule
Approach     & Client Computation                 & Server Computation                            \\ \midrule
Ours         & $O(mn+md) $                  & $O(m(n-d))$                           \\
\cite{BonawitzIKMMPRS17} &$ O(n^2+mn)$ & $O(m(n-d)+md(n-d))$ \\ \bottomrule
\end{tabular}
\end{table}

%
%

\subsubsection{Cloud Server-Side Performance}
We now examine the costs of securely aggregating model updates to produce an updated global model under the semi-honest adversary setting and active-adversary setting respectively.
The results are plotted in Fig. \ref{fig:cloud_agg_cost}.
As expected, the computation costs scales linearly with the fraction. 
%
%
%
It is revealed that our protocol with security against an active adversarial cloud server (i.e., computation integrity against the cloud server) incurs almost no overhead over the semi-honest setting ($1\times$, $1.005\times$, and $1.013\times$ over the MNIST, CIFAR, and CelebA models respectively).
%
%
Such minimal performance overhead is promised as no paging is required.

\begin{figure*}[!t]
\centering
\small
  \begin{minipage}[t]{0.3\linewidth}
    \centering{\includegraphics[width=\linewidth]{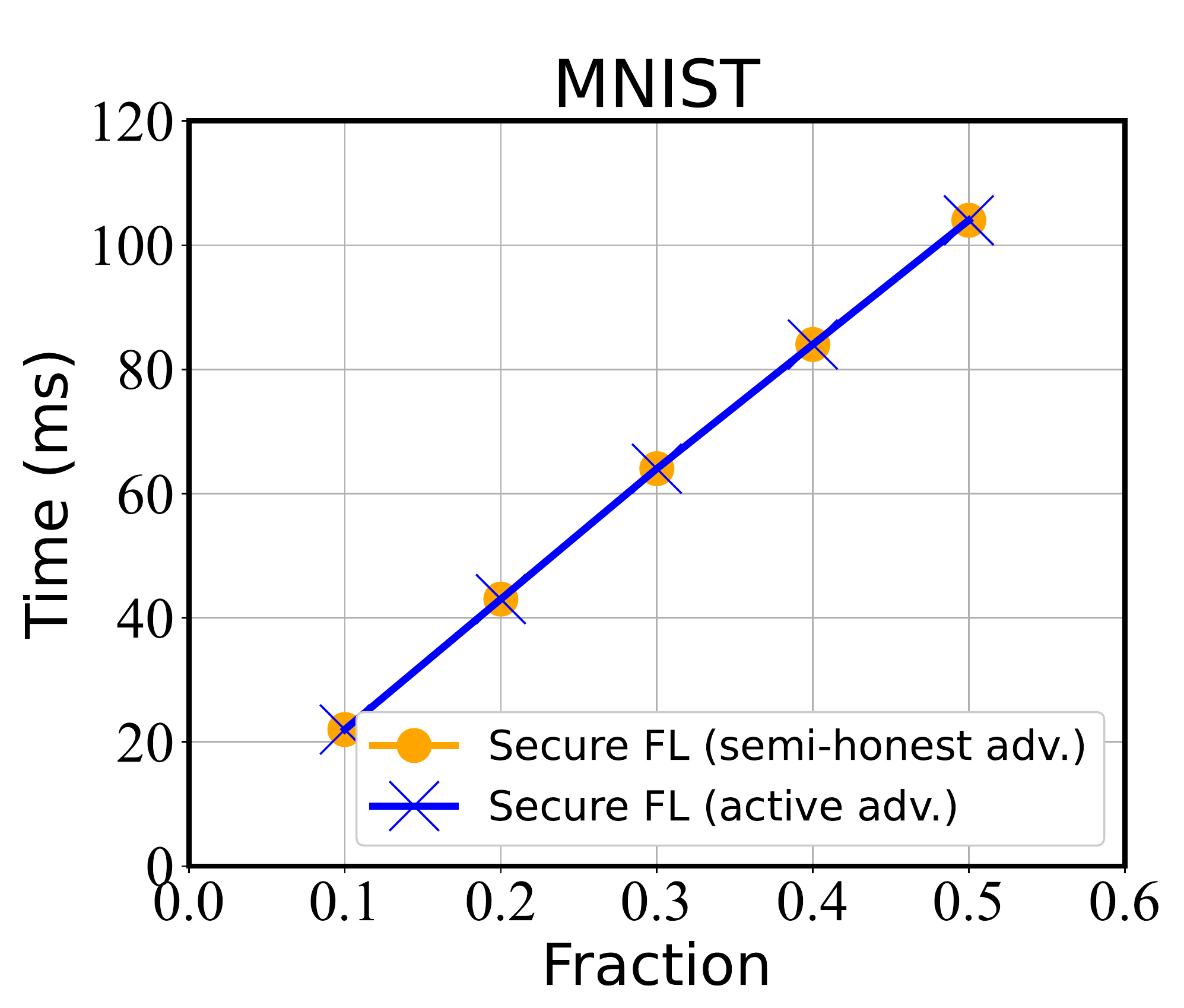}}
    
  \end{minipage}%
  \begin{minipage}[t]{0.3\linewidth}
    \centering{\includegraphics[width=\linewidth]{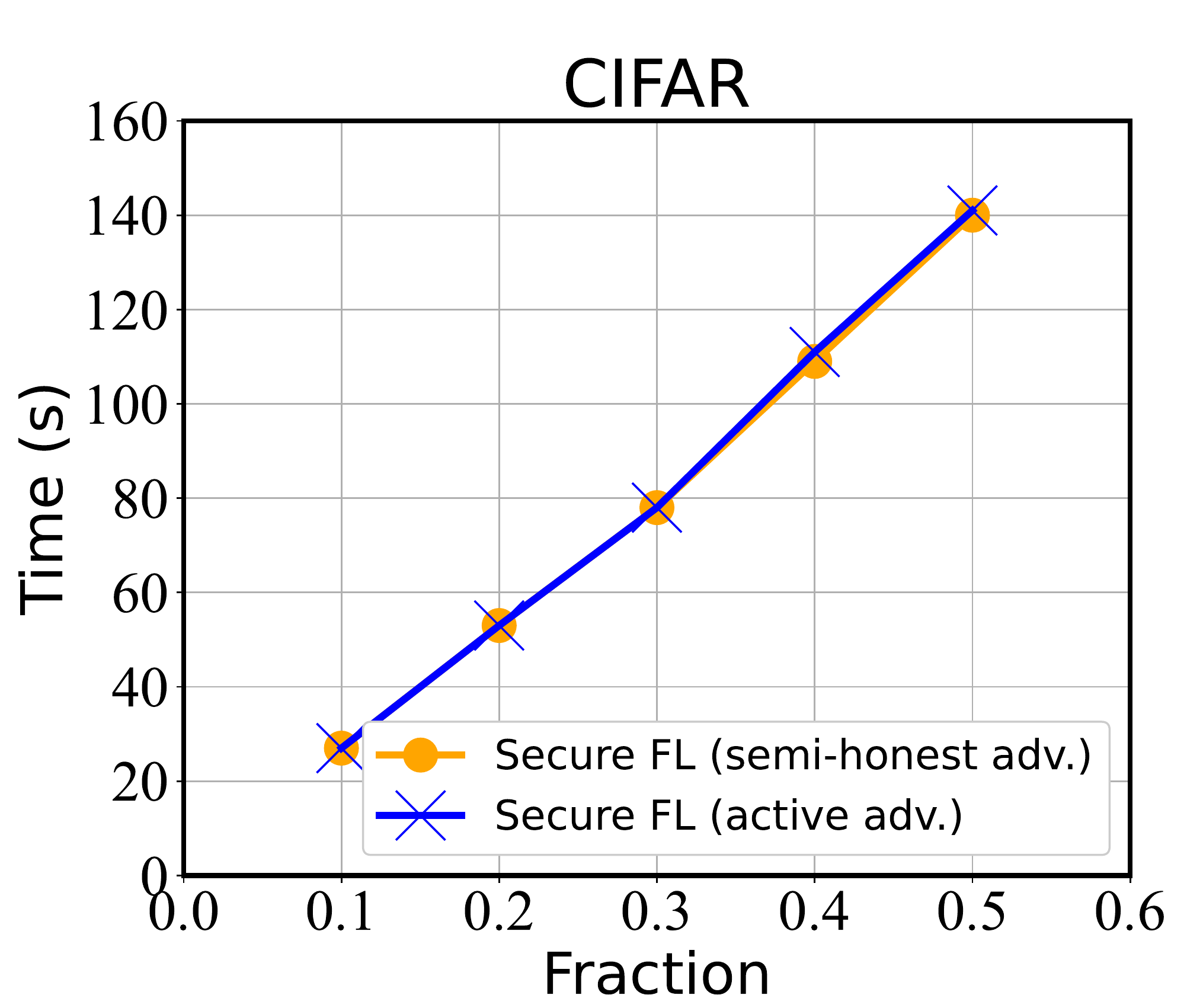}}

  \end{minipage}
    \begin{minipage}[t]{0.3\linewidth}
    \centering{\includegraphics[width=\linewidth]{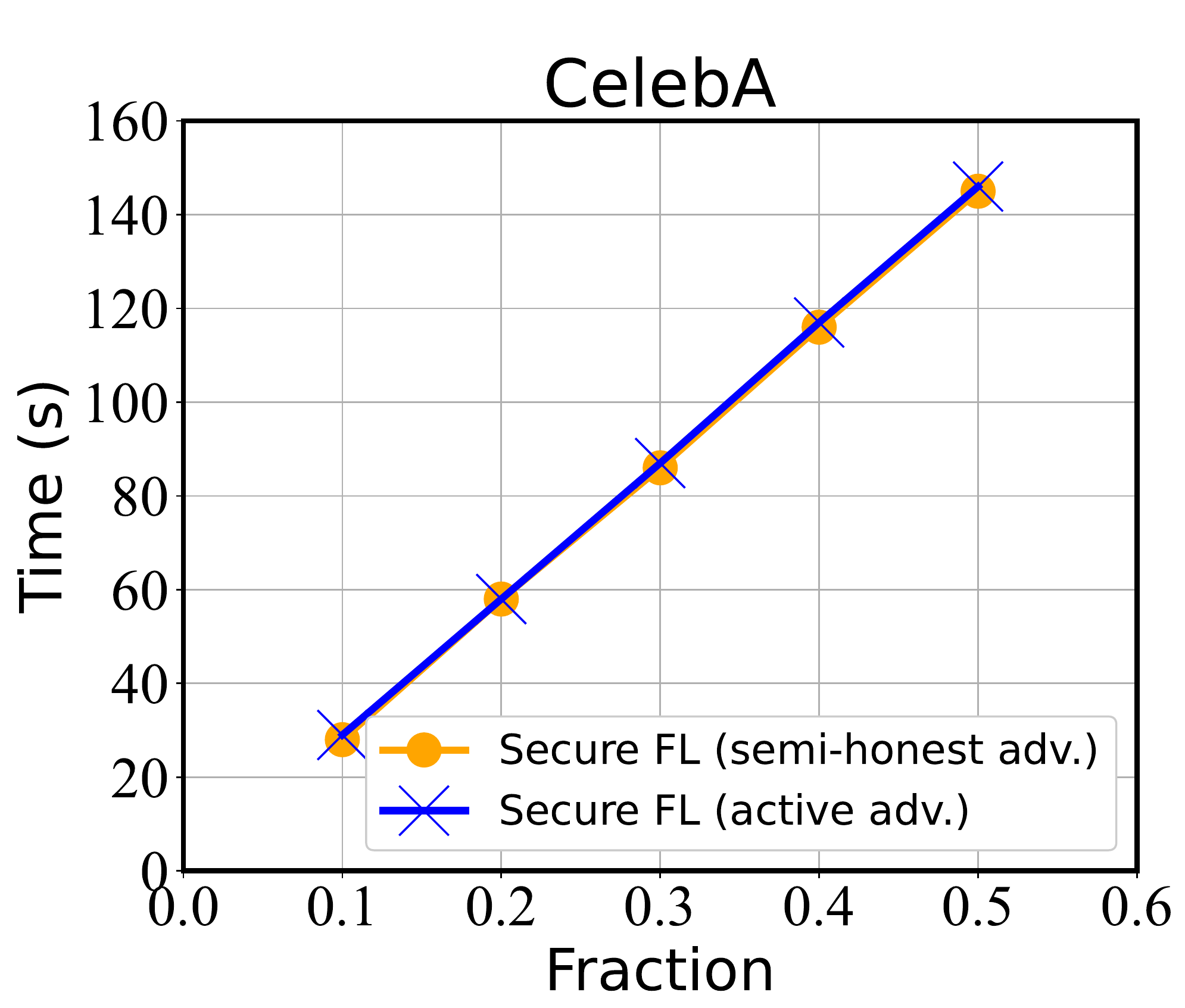}}
  \end{minipage}
  \caption{Cost of aggregating model updates at the cloud server, under different adversary settings and over different models}.
    \label{fig:cloud_agg_cost}
    \vspace{-10pt}
\end{figure*}

\subsubsection{Comparison with Prior Work}

We now make comparison with the most related prior work \cite{BonawitzIKMMPRS17} without heavy cryptographic operations.
Firstly, we compare the computational complexity on the client side and on the cloud server side respectively.
Suppose the dimension of each model update vector is $m$, and the number of clients being selected in a round of the federated learning procedure is $n$, and the number of dropped clients is $d$.
Overall the secure aggregation approach in our system leads to $O(mn+md)$ computation on the client side and $O(m(n-d))$ computation on the cloud server side.
In comparison, according to \cite{BonawitzIKMMPRS17}, their scheme leads to $O(n^2+mn)$ computation on the client side and $O(m(n-d)+md(n-d))$ on the cloud server side.
Table \ref{tbl:complexity} summarizes the comparison of the asymptotic computational complexity.

\begin{table}[t!]
\caption{Client Cost Comparison with Prior Work \cite{BonawitzIKMMPRS17}}
\label{tbl:client_cost_comparison_ccs17}
\begin{tabular}{@{}ccccccc@{}}
\toprule
\multirow{2}{*}{Dropouts} & \multirow{2}{*}{\begin{tabular}[c]{@{}c@{}}Client Cost\\ (ms)\end{tabular}} & \multicolumn{5}{c}{Fraction}          \\ \cmidrule(l){3-7} 
                          &                                                                             & 0.1   & 0.2   & 0.3   & 0.4   & 0.5   \\ \midrule
\multirow{2}{*}{0\%}      & Ours                                                                        & 0.95  & 1.09  & 1.2   & 1.35  & 1.49  \\ \cmidrule(l){2-7} 
                          & {[}7{]}                                                                     & 12.05 & 23.38 & 34.98 & 47.16 & 58.43 \\ \midrule
\multirow{2}{*}{10\%}     & Ours                                                                        & 16    & 28    & 42    & 56    & 69    \\ \cmidrule(l){2-7} 
                          & {[}7{]}                                                                     & 11.98 & 23.42 & 35.04 & 46.91 & 58.32 \\ \midrule
\multirow{2}{*}{20\%}     & Ours                                                                        & 28    & 59    & 82    & 108   & 134   \\ \cmidrule(l){2-7} 
                          & {[}7{]}                                                                     & 11.87 & 23.42 & 35.77 & 47.02 & 58.3  \\ \bottomrule
\end{tabular}
\end{table}

We note that the work \cite{BonawitzIKMMPRS17} does not present real machine learning based experiments.
To have empirical performance comparison with \cite{BonawitzIKMMPRS17}, we test their scheme\footnote{Implementation used: https://github.com/55199789/PracSecure} over the MNIST model to measure the client-side and server-side runtime costs, with varying dropout rates and fraction of selected clients per around.
Table \ref{tbl:client_cost_comparison_ccs17} gives the comparison of the client-side cost.
Our design is (up to $39\times$) more efficient than the work \cite{BonawitzIKMMPRS17} when the dropout rate is zero.
As the dropout rate increases, our design has higher client cost (limited to $2.3\times$), as each online client assists by computing blinding factors scaling to the number of drop-out clients.
%

\begin{table}[t!]
\caption{Server Cost Comparison with Prior Work \cite{BonawitzIKMMPRS17}}
\label{tbl:server_cost_comparison_ccs17}
\begin{tabular}{@{}ccccccc@{}}

\toprule
\multirow{2}{*}{Dropouts} & \multirow{2}{*}{\begin{tabular}[c]{@{}c@{}}Server Cost\\ (ms)\end{tabular}} & \multicolumn{5}{c}{Fraction}             \\ \cmidrule(l){3-7} 
                          &                                                                             & 0.1   & 0.2   & 0.3    & 0.4    & 0.5    \\ \midrule
\multirow{2}{*}{0\%}      & Ours                                                                        & 22    & 43    & 64     & 84     & 104    \\ \cmidrule(l){2-7} 
                          & {[}7{]}                                                                     & 5.74  & 11.26 & 16.32  & 22.59  & 29.66  \\ \midrule
\multirow{2}{*}{10\%}     & Ours                                                                        & 38    & 79    & 115    & 153    & 196    \\ \cmidrule(l){2-7} 
                          & {[}7{]}                                                                     & 13.53 & 41.15 & 83.53  & 141.98 & 220.5  \\ \midrule
\multirow{2}{*}{20\%}     & Ours                                                                        & 35    & 75    & 111    & 148    & 189    \\ \cmidrule(l){2-7} 
                          & {[}7{]}                                                                     & 19.83 & 64.17 & 138.89 & 238.24 & 366.59 \\ \bottomrule
\end{tabular}
\end{table}

Regarding the server-side cost, it is observed that as the dropout rate and the fraction of selected clients increase, our design (semi-honest adv. setting) consumes less computation. 
This is because the server in their scheme needs to perform reconstruction of secret keys over collected secret shares and re-compute masks to be subtracted from the aggregate sum over the online clients.
We emphasize that unlike \cite{BonawitzIKMMPRS17}, our system does not reveal the secret keys of drop-out clients, so their direct participation in future rounds is not affected. 
Meanwhile, our system can also provide much stronger security (computation integrity) against the server with minimal overhead, as demonstrated above.

\section{Conclusion}
\label{sec:conclusion}

In this paper, we present a system design for federated learning, which allows clients to provide obscured model updates while aggregation can still be supported.
Our system first departs from prior works by building on a cherry-picked cryptographic aggregation protocol, which promises the advantages of lightweight encryption and aggregation as well as the ability to handle drop-out clients without exposing their secret keys.
For higher communication efficiency, our system also adapts the latest advancements in quantization techniques for compressing individual model updates.
Furthermore, our system also provides security beyond the common semi-honest adversary setting, ensuring the computation integrity at the cloud server. 
We conduct an extensive evaluation over popular benchmark datasets, and the results validate the practical performance of our system.

\section*{Acknowledgement}
This work was supported in part by the Guangdong Basic and Applied Basic Research Foundation under Grant 2021A1515110027, in part by the Shenzhen High-Level Talents Research Start-up Fund, in part by the Australian Research Council (ARC) Discovery Projects under Grants DP200103308 and DP180103251, in part by a Monash-Data61 collaborative research project (Data61 CRP43), in part by the Research Grants Council of Hong Kong under Grants CityU 11217819, 11217620, 11218521, N\_CityU139/21, R6021-20F, and RFS2122-1S04, in part by Shenzhen Municipality Science and Technology Innovation Commission under Grant SGDX20201103093004019, and in part by the National Natural Science Foundation of China under Grant 61572412.

\balance
\bibliographystyle{IEEEtran}

\bibliography{references}

\end{document}